\documentclass[acmsmall,screen,nonacm]{acmart}

\setcopyright{none}

\usepackage{mathtools}
\usepackage{enumitem}
\usepackage{booktabs}
\usepackage{tabularx}
\usepackage{array}
\usepackage{longtable}
\usepackage{microtype}
\usepackage{upquote}
\usepackage{xurl}
\usepackage[most]{tcolorbox}
\usepackage{tikz}
\usetikzlibrary{arrows.meta,positioning,calc,fit,shapes.geometric}
\allowdisplaybreaks
\definecolor{ReaderBlue}{HTML}{EAF3FA}
\definecolor{ReaderBlueFrame}{HTML}{5C86A6}
\definecolor{ReaderGreen}{HTML}{EDF7EF}
\definecolor{ReaderGreenFrame}{HTML}{5B8C63}
\definecolor{ReaderGold}{HTML}{FFF7E6}
\definecolor{ReaderGoldFrame}{HTML}{A77A22}

\newtcolorbox{readerbox}[1]{
  enhanced,breakable,
  colback=ReaderBlue,colframe=ReaderBlueFrame,
  boxrule=0.6pt,arc=1.5mm,
  left=1.5mm,right=1.5mm,top=1mm,bottom=1mm,
  fonttitle=\bfseries,title={#1}}
\newtcolorbox{intuitionbox}[1]{
  enhanced,breakable,
  colback=ReaderGreen,colframe=ReaderGreenFrame,
  boxrule=0.6pt,arc=1.5mm,
  left=1.5mm,right=1.5mm,top=1mm,bottom=1mm,
  fonttitle=\bfseries,title={#1}}
\newtcolorbox{warningbox}[1]{
  enhanced,breakable,
  colback=ReaderGold,colframe=ReaderGoldFrame,
  boxrule=0.6pt,arc=1.5mm,
  left=1.5mm,right=1.5mm,top=1mm,bottom=1mm,
  fonttitle=\bfseries,title={#1}}
\newtcolorbox{roadmapbox}[1]{
  enhanced,breakable,
  colback=ReaderBlue,colframe=ReaderBlueFrame,
  boxrule=0.8pt,arc=1.5mm,
  left=2mm,right=2mm,top=1.2mm,bottom=1.2mm,
  fonttitle=\bfseries,title={Roadmap: #1}}
\newtcolorbox{detailbox}[1]{
  enhanced,breakable,
  colback=white,colframe=black!38,
  boxrule=0.5pt,arc=1.2mm,
  left=1.7mm,right=1.7mm,top=1mm,bottom=1mm,
  fonttitle=\bfseries,title={#1}}
\newtcolorbox{checklistbox}[1]{
  enhanced,breakable,
  colback=ReaderGold,colframe=ReaderGoldFrame,
  boxrule=0.7pt,arc=1.5mm,
  left=1.8mm,right=1.8mm,top=1mm,bottom=1mm,
  fonttitle=\bfseries,title={Proof check: #1}}

\theoremstyle{definition}
\newtheorem{example}{Example}[section]
\theoremstyle{remark}
\newtheorem{remark}{Remark}[section]

\DeclareMathOperator{\dc}{dc}
\DeclareMathOperator{\sdc}{sdc}
\DeclareMathOperator{\rank}{rank}
\DeclareMathOperator{\adj}{adj}
\DeclareMathOperator{\tr}{tr}

\DeclareMathOperator{\Sing}{Sing}

\DeclareMathOperator{\Con}{Con}
\DeclareMathOperator{\supp}{supp}
\DeclareMathOperator{\Mat}{Mat}
\DeclareMathOperator{\im}{im}

\newcommand{\PP}{\mathbb{P}}
\newcommand{\CC}{\mathbb{C}}
\newcommand{\AAff}{\mathbb{A}}
\newcommand{\OO}{\mathcal{O}}
\newcommand{\Gfam}{\mathcal{G}}
\newcommand{\pdeg}{\delta}
\newcommand{\dcb}{\overline{\dc}}
\newcommand{\sdcb}{\overline{\sdc}}
\newcommand{\Gone}{\Gamma_{\!1}}

\usepackage{fancyvrb}
\providecommand{\leanRepositoryURL}{\url{https://github.com/karthik202321/power-sum-dc-lean}}
\providecommand{\leanArchiveDOI}{\url{https://doi.org/10.5281/zenodo.23274251}}
\providecommand{\leanPaperURL}{\url{https://arxiv.org/abs/2606.13628}}

\begin{document}

\title[A fully Lean-verified quadratic border determinantal lower bound]{A fully Lean-verified quadratic lower bound on the border determinantal complexity of \texorpdfstring{$x_1^n+\cdots+x_n^n$}{x1\textasciicircum n + ... + xn\textasciicircum n} over \texorpdfstring{$\CC$}{C}, including the symmetric and exact models}
\subtitle{Proof by conormal specialization}

\author{Karthik Sheshadri}
\affiliation{%
  \institution{Independent researcher}
  \city{San Jose}
  \state{California}
  \country{USA}}
\email{karthiksheshadri217@gmail.com}
\renewcommand{\shortauthors}{K. Sheshadri}

\begin{abstract}
We present the first end-to-end Lean-verified superlinear lower bound on border determinantal complexity.  Let $F_n=x_1^n+\cdots+x_n^n$.  We prove, over $\CC$, that
\[
  \dcb(F_n)\ge \frac{(n-1)^2}{4e},
  \qquad
  \sdcb(F_n)\ge \frac{(n-1)^2}{2e},
\]
for the ordinary and symmetric border determinantal models, and therefore the same bounds for exact determinantal complexity.  Together with known $O(n^2)$ constructions, both border quantities are $\Theta(n^2)$.

The proof compares one enumerative invariant on two sides of a degeneration.  On the determinant side, the $(n-2)$nd multidegree of the multiplicity-one Gauss graph of any size-$m$ affine-linear determinant is bounded by an explicit multihomogeneous B\'ezout number; no smoothness, irreducibility, or squarefreeness is assumed.  On the target side, a flat limit of these effective Gauss-graph cycles contains the conormal variety of the projective Fermat cone.  Its relevant multidegree is the classical dual degree $n(n-1)^{n-2}$.  Effectivity gives the inequality in the correct direction, and an $(n-1)$st root yields the quadratic bound.

All four bounds are fully verified in Lean~4 with Mathlib.  The Lean theorems assume only $n\ge3$ and a literal determinantal representation or literal border membership, and they depend on no axioms beyond Lean's foundational ones.  The formal proof follows a localized version of the argument presented here: where the text compares global cycles, it compares local intersection lengths of finite slices of those cycles.

Lean sources: \leanRepositoryURL.\\
Archived release: \leanArchiveDOI.

The exposition begins with a geometric dictionary, a one-line proof chain, and a dependency diagram; the longer incidence and flatness arguments then verify those high-level steps.
\end{abstract}

\keywords{border determinantal complexity, algebraic complexity, conormal variety, polar degree, Gauss map, geometric complexity theory, Lean, formal verification}

\maketitle

\begin{readerbox}{How the paper is organized}
The paper is written in three layers.  The first layer --- the Introduction and
Section~\ref{sec:guide} --- explains the whole argument without requiring
scheme theory or intersection theory.  The second layer consists of the blue,
green, and gold boxes placed immediately before the technical proofs; each box
states the goal of the next calculation, the input it uses, and the output that
will be needed later.  The third layer is the complete written proof.  The boxes
are explanatory only: every mathematical implication is also proved in the
ordinary theorem--lemma text.  A fourth layer, Appendix~\ref{app:lean-verification},
reports the complete Lean verification of the exact and border numerical
endpoints and explains how to reproduce it independently.
\end{readerbox}

\section{Introduction}\label{sec:intro}

\subsection{The question and the answer}

For a polynomial $f(x_1,\ldots,x_n)$, its determinantal complexity is the smallest $m$ for which
\[
  f(x)=\det\!\left(A_0+x_1A_1+\cdots+x_nA_n\right)
\]
with constant $m\times m$ matrices $A_i$.  Let
\[
  D_m=\left\{\det\!\left(A_0+\sum_{i=1}^n x_iA_i\right):
  A_i\in\CC^{m\times m}\right\}.
\]
The \emph{border determinantal complexity} is
\[
  \dcb(f)=\min\{m:f\in\overline{D_m}\},
\]
where the closure is Zariski closure (equivalently, over $\CC$, Euclidean closure for this constructible set).  Thus $f$ need only be a limit of size-$m$ determinants.  Restricting the matrices to be symmetric defines $\sdcb(f)$.

We study the diagonal power sum
\[
  F_n=x_1^n+\cdots+x_n^n.
\]
Exact determinantal lower bounds for explicit polynomials have been stubbornly close to linear in the number of variables, while the standard construction for $F_n$ has quadratic size.  The following result closes that gap up to a constant even in the border model.

\begin{theorem}[Main theorem: ordinary model]\label{thm:border}
For every $n\ge3$, over $\CC$,
\[
  \dcb(F_n)\ \ge\ \frac{(n-1)^2}{4e}.
\]
\end{theorem}

\begin{theorem}[Main theorem: symmetric model]\label{thm:symborder}
For every $n\ge3$, over $\CC$,
\[
  \sdcb(F_n)\ \ge\ \frac{(n-1)^2}{2e}.
\]
\end{theorem}

Known constructions give $\dc(F_n)=O(n^2)$ and $\sdc(F_n)=O(n^2)$ \cite{KumarVolk21,SheshadriSDC}.  Since border complexity is at most exact complexity, the two theorems imply that both border quantities are $\Theta(n^2)$.

\begin{corollary}[Exact bounds]\label{cor:exact}
For every $n\ge3$,
\[
  \dc(F_n)\ge\frac{(n-1)^2}{4e},
  \qquad
  \sdc(F_n)\ge\frac{(n-1)^2}{2e}.
\]
\end{corollary}

\subsection{End-to-end Lean verification}\label{sec:lean-overview}
Theorems~\ref{thm:border} and~\ref{thm:symborder} and Corollary~\ref{cor:exact} are verified end to end in Lean~4 with Mathlib.  To our knowledge, they are the first determinantal or border determinantal lower bounds verified end to end in a proof assistant.  The Lean statements (Appendix~\ref{app:lean-statements}) assume only $n\ge3$ and either a literal determinantal representation of $F_n$ or literal membership of $F_n$ in the closure of the size-$m$ determinants.  They depend on no axioms beyond Lean's three foundational ones (\texttt{propext}, \texttt{Classical.choice} and \texttt{Quot.sound}).  Every intermediate input---B\'ezout's theorem for isolated zeros with excess, the generic flags, the degree computation and the one-parameter family---is constructed or proved inside the development, not assumed.

The Lean proof follows a localized version of the argument in this paper.  The text compares global cycles: the conormal variety of the Fermat cone and the limit of the Gauss-graph cycles.  The Lean proof compares finite slices of them.  A \emph{local Gauss section} of a hypersurface $f=0$ is the finite scheme cut out by $f=0$, the $n-2$ polar equations $\sum_i\ell_{t,i}\,\partial_if=0$ $(1\le t\le n-2)$ and one affine chart equation $a\cdot x=1$.  It is the zero-dimensional slice of the Gauss graph that the multidegree $\pdeg_{n-2}$ counts (Section~\ref{sec:multidegree}), and the lengths of its local rings are the multiplicities of its points.  Comparing these lengths across the degeneration gives the inequalities of \eqref{eq:one-line-proof} in local form.  The verification therefore covers the stated theorems and every lemma in their dependency chain, but not every sentence about global cycles in the exposition.  Appendix~\ref{app:lean-verification} gives the exact Lean statements, their correspondence with the argument, and build and audit instructions.

\begin{readerbox}{Formalization and archival links}
Lean sources: \leanRepositoryURL.\\
Archived release (Lean sources, audits and this paper): \leanArchiveDOI.\\
arXiv record of this paper: \leanPaperURL.
\end{readerbox}

\subsection{The proof in one displayed line}

Let $X\subset\PP^n$ be the projective cone defined by the same form $F_n$ (now viewed as independent of $x_0$).  If $F_n\in\overline{D_m}$, that is, if $F_n$ is a limit of size-$m$ determinants, we construct a nearby determinantal fibre with multiplicity-one Gauss-graph cycle $\Gamma_c$.  The proof is the chain
\begin{equation}\label{eq:one-line-proof}
\underbrace{n(n-1)^{n-2}}_{\text{Fermat dual degree}}
=\pdeg_{n-2}(\Con(X))
\ \le\ 
\underbrace{\pdeg_{n-2}(\Gamma_c)}_{\text{conserved effective cycle}}
\ \le\ 
\underbrace{B(m,n)}_{\text{kernel incidence}}
\ \le\ 
\left(\frac{4em}{n-1}\right)^{n-1}.
\end{equation}
The first equality is classical projective duality.  The first inequality is the specialization step: the conormal variety of the Fermat cone occurs as an effective component of the limiting Gauss-graph cycle.  The second inequality is a determinant-specific B\'ezout count.  The final inequality is elementary.  Since the left side is at least $(n-1)^{n-1}$, taking $(n-1)$st roots gives $m\ge(n-1)^2/(4e)$.  In the symmetric model the determinant-side estimate improves by a factor of two in the base of the power.

\begin{roadmapbox}{the five proof obligations behind (\ref*{eq:one-line-proof})}
\begin{enumerate}[leftmargin=2.2em,itemsep=2pt]
\item \textbf{Turn border membership into a family.}  From
$F_n\in\overline{D_m}$, construct a curve of polynomials $G_c$ with
$G_0=F_n$ and with every generic $G_c$ an exact size-$m$ determinant.
This is Lemma~\ref{lem:normalform}.
\item \textbf{Package tangent data in a flat family.}  For each generic
fibre, form the multiplicity-one Gauss graph $\Gamma_c$.  Keep the
components that move with the parameter and close them up to a proper flat
family $W\to C$.  This is Section~\ref{sec:family}.
\item \textbf{Show that the target survives.}  Prove that the conormal
variety of the Fermat cone is an irreducible component of the effective
special fibre $[W_0]$.  This is Proposition~\ref{prop:containment}.
\item \textbf{Compute the target-side count.}  Show that the relevant
multidegree of the Fermat cone is $n(n-1)^{n-2}$.  This is
Proposition~\ref{prop:coneshift}.
\item \textbf{Bound the determinant-side count.}  Lift generic Gauss
points to left and right kernel lines and count the lifts by a
multihomogeneous B\'ezout number.  This is
Theorem~\ref{thm:detlemma}.
\end{enumerate}
Once these five statements are available, the main theorem is the arithmetic
in Section~\ref{sec:proofs}.
\end{roadmapbox}

\begin{figure}[H]
\centering
\begin{tikzpicture}[
  >=Latex,
  node distance=7mm and 8mm,
  every node/.style={font=\footnotesize,align=center},
  box/.style={draw,rounded corners,inner sep=5pt,text width=3.0cm},
  special/.style={box,fill=ReaderGreen},
  generic/.style={box,fill=ReaderBlue},
  conclusion/.style={box,fill=ReaderGold}
]
\node[conclusion] (border) {$F_n\in\overline{D_m}$\\border size $m$};
\node[box,below=of border] (curve) {Curve family $G_c$\\$G_0=F_n$, generic $G_c=\det A_c$};
\node[special,below left=10mm and 18mm of curve] (special) {Special fibre\\Fermat cone component};
\node[generic,below right=10mm and 18mm of curve] (generic) {Generic fibre\\affine-linear determinant};
\node[special,below=of special] (capture) {$\Con(X)\subset\supp W_0$\\effectivity + conservation};
\node[generic,below=of generic] (bezout) {Lift Gauss points to kernels\\multihomogeneous B\'ezout};
\node[conclusion,below=16mm of curve] (compare) {Compare the same multidegree\\$n(n-1)^{n-2}\le B(m,n)$};
\draw[->] (border) -- (curve);
\draw[->] (curve) -- (special);
\draw[->] (curve) -- (generic);
\draw[->] (special) -- (capture);
\draw[->] (generic) -- (bezout);
\draw[->] (capture) -- (compare);
\draw[->] (bezout) -- (compare);
\end{tikzpicture}
\caption{The proof separates into a target-side lower count and a determinant-side upper count.  Flat specialization lets the two counts be compared.}
\Description{A flow diagram: border membership gives a curve family, which splits into a special Fermat-cone branch and a generic determinant branch.  The branches provide lower and upper multidegree bounds that are compared at the end.}
\label{fig:proof-map}
\end{figure}

\subsection{The two counts and the bridge between them}

The proof compares a \emph{lower count} attached to the target polynomial with
an \emph{upper count} attached to any size-$m$ determinant.  The quantities
are deliberately chosen to be the same multidegree.

\begin{center}
\small
\begin{tabularx}{\textwidth}{@{}p{0.19\textwidth}X X@{}}
\toprule
& Target side & Determinant side \\
\midrule
Geometric object & Conormal variety of the Fermat cone &
Multiplicity-one Gauss graph of $\det\widehat A$ \\
Quantity counted & $\pdeg_{n-2}$, equal to $n(n-1)^{n-2}$ &
The same $\pdeg_{n-2}$, at most $B(m,n)$ \\
Reason for the bound & Classical dual degree of a smooth Fermat hypersurface &
Kernel lift plus multihomogeneous B\'ezout \\
Bridge & \multicolumn{2}{c}{A flat effective limit contains the target conormal, so target count $\le$ generic count.} \\
\bottomrule
\end{tabularx}
\end{center}

The central design choice is therefore not merely to find a large invariant of
$F_n$, but to find one that can be compared across a degeneration without
changing sign or losing multiplicity.

\subsection{Why the border step is not automatic}

The exact-representation argument reads a polar or conormal degree from a fixed determinant.  Such degrees are not closed conditions on polynomial space: under degeneration, isolated solutions can merge into singular strata or excess components.  A local singularity invariant at the vertex of the Fermat cone can even jump \emph{up} in the special fibre, so the naive semicontinuity direction is wrong.

The repair is to specialize a \emph{global effective cycle}, not a local number.  The Gauss graphs of the generic fibres have a flat limit $W_0$.  Their multidegrees are conserved, and every irreducible component of $W_0$ has a nonnegative coefficient.  Once $\Con(X)$ is shown to be one of those components, its multidegree is automatically at most the conserved total.  Extra components are harmless because they add rather than subtract.

A second issue is non-reducedness.  The degree-$m$ homogenization of the special fibre is $x_0^{m-n}F_n$, and nearby determinants may also have repeated factors.  We therefore use the \emph{multiplicity-one Gauss graph} $\Gone(F)$: it is the sum of the conormal varieties of those irreducible components on which $F$ vanishes with multiplicity exactly one.  Equivalently, it is the part of the hypersurface seen by the nonzero differential $dF$.  This is both the cycle controlled by the determinant-side incidence and the cycle that captures the Fermat component under specialization.

\subsection{The determinant-side estimate}

The only long algebraic calculation in the proof is isolated in the following theorem.  Its role is simple: a size-$m$ determinant cannot have an arbitrarily large $(n-2)$nd conormal multidegree.

\begin{theorem}[Determinantal conormal bound]\label{thm:detlemma}
Let $n\ge2$, $m\ge1$, and let $A(x)=A_0+\sum_{i=1}^n x_iA_i$ be an $m\times m$ matrix of affine-linear forms over $\CC$.  Let $\widehat A=x_0A_0+\sum_i x_iA_i$ and suppose $F=\det\widehat A\not\equiv0$.  Then
\begin{enumerate}[label=(\roman*),leftmargin=2.4em]
\item
\[
  \pdeg_{n-2}(\Gone(F))\le B(m,n),
\]
where
\[
  B(m,n)=\big[x^nu^{m-1}v^{m-1}\big]\,
  x(x+u)^m(x+v)^{m-1}(u+v)^{n-2}
\]
and equivalently
\[
  B(m,n)=\sum_{i=1}^{n-1}
  \binom{m}{i}\binom{m-1}{n-1-i}\binom{n-2}{i-1}.
\]
\item If all $A_i$ are symmetric, then
\[
  \pdeg_{n-2}(\Gone(F))\le
  2^{n-2}\binom{m}{n-1}.
\]
\end{enumerate}
No smoothness, irreducibility, squarefreeness, or genericity assumption is made on $F$ or on the matrices $A_i$.
\end{theorem}

\begin{readerbox}{What to remember from Theorem~\ref*{thm:detlemma}}
The exact coefficient formula is needed only at the end.  Conceptually,
\[
  B(m,n)\le\left(\frac{4em}{n-1}\right)^{n-1},
  \qquad
  2^{n-2}\binom{m}{n-1}\le\left(\frac{2em}{n-1}\right)^{n-1}.
\]
The proof lifts each Gauss point to kernel vectors and counts the lifted points by B\'ezout.
\end{readerbox}

\subsection{Context and relation to prior work}

Determinantal complexity is a central model in algebraic complexity, and border complexity is the natural orbit-closure version used in geometric complexity theory.  For exact determinantal complexity of $F_n$, the prior published lower bound was $1.5n-3$ \cite{KumarVolk21}, following earlier singular-locus methods \cite{AlperBogartVelasco17}.  Mignon and Ressayre's quadratic lower bound for the permanent \cite{MignonRessayre04} extends to the border through a dual-\emph{dimension} obstruction \cite{LMR13}; measured in the permanent's number of variables, that bound is linear.  To our knowledge, Theorems~\ref{thm:border} and~\ref{thm:symborder} give the first border determinantal lower bounds for an explicit $n$-variate family that are superlinear in $n$.

The paper supersedes the exact lower bounds in the companion manuscripts \cite{SheshadriDC,SheshadriSDC}.\footnote{A shorter proof of a quadratic bound in the exact ordinary model was later given in~\cite{KumarVolk26}; by its authors' account, its outline came out of parsing parts of the proof in an earlier version of this paper.  It is cited here as a reading aid for that case only.  No argument in this paper uses it, and it does not treat the border or symmetric models.}  The determinant-side incidence originates there, but Theorem~\ref{thm:detlemma} removes their smoothness requirement.  The new ingredient is the effective conormal-cycle specialization that transfers the lower bound to border complexity.

\subsection{How to read the paper}

Section~\ref{sec:guide} gives a geometric dictionary and a concrete low-dimensional picture.  Section~\ref{sec:prelim} records the formal definitions.  Section~\ref{sec:detlemma} proves the determinant-side black box; it may be skipped on a first reading.  Sections~\ref{sec:normalform}--\ref{sec:capture} construct and analyze the degeneration, Section~\ref{sec:coneshift} computes the target-side multidegree, and Section~\ref{sec:proofs} combines the inequalities.  Readers who want only what is machine-checked can read Section~\ref{sec:lean-overview} and Appendix~\ref{app:lean-statements}.

\section{A geometric guide for TCS readers}\label{sec:guide}

This section translates the intersection-theoretic language into the finite
counts used by the proof.  Nothing in this section replaces a later proof; its
purpose is to make the objects and the direction of every inequality familiar
before the technical construction begins.

\subsection{Why projective space and homogenization appear}

An affine size-$m$ determinant has degree at most $m$:
\[
  G(x)=\det\!\left(A_0+\sum_{i=1}^n x_iA_i\right).
\]
Introduce a homogenizing variable $x_0$ and set
\[
  \widehat A(x_0,x)=x_0A_0+\sum_{i=1}^n x_iA_i.
\]
Then $\det\widehat A$ is a homogeneous degree-$m$ form, and on the affine
chart $x_0=1$ it equals $G$.  Passing to projective space has three concrete
advantages.

\begin{enumerate}[leftmargin=2.2em,itemsep=2pt]
\item Points that move to infinity remain inside the compact projective
ambient space; properness later guarantees that limiting Gauss points do not
simply disappear.
\item Tangent hyperplanes themselves form a projective space
$(\PP^n)^\vee$, so a point and its tangent hyperplane naturally live in the
product $\PP^n\times(\PP^n)^\vee$.
\item Intersection numbers with generic linear spaces are finite and stable;
these are the multidegrees used in the proof.
\end{enumerate}

The target $F_n$ has degree $n$, whereas the approximating determinant has
formal homogeneous degree $m$.  Homogenizing at degree $m$ therefore gives
\[
  \widehat G_0=x_0^{m-n}\widetilde F_n.
\]
The factor $x_0^{m-n}$ is the padding needed to compare degree-$m$ forms.  It
may make the special divisor non-reduced at infinity, but it does not increase
the multiplicity of the Fermat factor itself.  This distinction is used
throughout the specialization argument.

\subsection{Gauss graphs and conormal varieties}

Let $S=V(f)\subset\PP^n$ be a smooth hypersurface.  At a point $x\in S$, the gradient $df(x)$ defines the tangent hyperplane.  The map
\[
  x\longmapsto (x,[df(x)])
\]
places each point together with its tangent hyperplane in
$\PP^n\times(\PP^n)^\vee$.  The closure of this graph is the conormal variety $\Con(S)$.

If $f$ is reducible or non-reduced, the differential vanishes identically on any component that occurs with multiplicity at least two.  The cycle $\Gone(f)$ keeps precisely the conormal varieties of the multiplicity-one components.  This is not an arbitrary convention: it is exactly the part of the divisor visible in the open set where $df\ne0$.

There are two useful elementary examples.

\begin{example}[One simple component and one repeated component]
Let $F=g h^2$, with $g$ and $h$ distinct reduced irreducible forms.  At a
generic point of $V(g)$ away from $V(h)$,
\[
  dF=h^2\,dg,
\]
which is nonzero, so the conormal of $V(g)$ is present in $\Gone(F)$.  At a
generic point of $V(h)$, every term in $dF$ contains a factor $h$, so
$dF=0$; the repeated component $V(h)$ is absent.  Thus
$\Gone(F)=\Con(V(g))$.
\end{example}

\begin{example}[Intersections do not remove multiplicity-one components]
Let $F=gh$ with $g,h$ reduced and distinct.  At a generic point of $V(g)$
away from $V(h)$, $dF=h\,dg\ne0$, and similarly on $V(h)$.  At the
intersection $V(g)\cap V(h)$ the differential may vanish, but that lower-
dimensional set does not matter: taking the closure of the open Gauss graph
recovers both conormal varieties.  Hence
$\Gone(F)=\Con(V(g))+\Con(V(h))$ as a cycle.
\end{example}

The word \emph{cycle} is important.  It means a formal sum of irreducible
subvarieties with integer coefficients.  Here every visible conormal component
is assigned coefficient one.  Later, a special fibre can contain the same
component with a larger positive coefficient; effectivity ensures that such
multiplicities help rather than hurt the desired inequality.

\subsection{What a multidegree counts}\label{sec:multidegree}

An $(n-1)$-dimensional cycle $C\subset\PP^n\times(\PP^n)^\vee$ has numbers
\[
  \pdeg_k(C)=\deg\big(C\cdot h_1^{n-1-k}h_2^k\big),
  \qquad 0\le k\le n-1.
\]
Enumeratively, $\pdeg_k(C)$ is the number of points of $C$ satisfying
$n-1-k$ generic linear conditions on the primal point and $k$ generic linear conditions on the dual hyperplane.  The dimension check is exact:
$C$ has dimension $n-1$, and the total number of imposed conditions is
$(n-1-k)+k=n-1$.  For generic choices the intersection is therefore
zero-dimensional, transverse, and supported away from any prescribed lower-
dimensional bad locus.

We use only $k=n-2$.  Thus a counted point must satisfy
\begin{enumerate}[leftmargin=2.2em,itemsep=1pt]
\item one generic linear condition on the point $x\in\PP^n$; and
\item $n-2$ generic linear conditions on its tangent hyperplane
$[df(x)]\in(\PP^n)^\vee$.
\end{enumerate}
For $n=3$, this is especially concrete: one plane condition on the point and
one hyperplane condition on the tangent plane.

The adjective \emph{generic} does not mean that the determinant or the family
is generic.  The input matrix may be completely special.  Only the auxiliary
linear spaces used to read the intersection number are chosen after the input,
outside a proper closed set of bad choices.

\begin{intuitionbox}{Why multidegrees are suitable for specialization}
They measure the whole fibre cycle.  In the proper flat family inside our fixed biprojective ambient space, every multidegree is conserved.  If the special fibre splits into several effective components, each component contributes a nonnegative amount to every multidegree.
\end{intuitionbox}

\subsection{Effective cycles and flat families in plain language}

An effective cycle has the form
\[
  a_1V_1+\cdots+a_rV_r,\qquad a_i\in\mathbb Z_{\ge0}.
\]
Every multidegree is additive:
$\pdeg_k(\sum_i a_iV_i)=\sum_i a_i\pdeg_k(V_i)$, and each summand is
nonnegative.  Consequently, if a special fibre contains a desired component
$V_1$ with coefficient at least one, then
\[
  \pdeg_k(V_1)\le \pdeg_k\!\left(\sum_i a_iV_i\right).
\]
This elementary inequality is the sign mechanism of the border proof.

Flatness is the geometric condition that makes the \emph{total} fibre cycle
specialize without losing intersection number.  A useful one-variable analogy
is a polynomial whose roots move with a parameter.  Two simple roots may
coalesce into a double root, so the number of distinct roots drops, but the
number counted with multiplicity is conserved.  In higher-dimensional
families, the extra multiplicity may be carried by an entire component over a
singular stratum; the cycle formalism records it automatically.

\subsection{A model degeneration: where the missing count goes}

For a smooth plane cubic, the dual curve has degree six.  When the cubic degenerates to a nodal cubic, the reduced conormal contributes only four generic tangencies; the remaining multiplicity two is supported over the node.  The total specialized cycle still has degree six:
\[
  6=4+2.
\]
This is the picture to keep in mind.  The visible conormal of the special hypersurface can have smaller degree than a generic fibre, but the deficit reappears as effective excess.  Our proof only needs the inequality ``one special component $\le$ conserved total.''

\subsection{A fully worked low-dimensional picture: $n=3$}

Take $F_3=x^3+y^3+z^3$.  Its projective zero set
$Z\subset\PP^2$ is a smooth plane cubic, whose dual curve has degree
$3(3-1)=6$.  The projective cone $X\subset\PP^3$ over $Z$ has ruling
lines joining the vertex to the points of $Z$.  The relevant invariant is
$\pdeg_1(\Con(X))$:
\begin{itemize}[leftmargin=2em,itemsep=1pt]
\item one generic condition is imposed on the point of the cone;
\item one generic condition is imposed on the tangent plane;
\item the tangent-plane condition selects six tangencies on the base cubic;
\item the point condition then cuts the corresponding six ruling lines once.
\end{itemize}
Hence $\pdeg_1(\Con(X))=6$.

On the determinant side the general coefficient simplifies to
\[
  B(m,3)=3\binom{m}{2}.
\]
Thus the exact comparison in this small case reads
$6\le3\binom{m}{2}$, which forces $m\ge3$.  This is not the asymptotic
content of the theorem, but it displays every geometric ingredient without
high-dimensional notation.  Under a nodal degeneration, the same total six
may split as four ordinary tangencies plus multiplicity two over the node:
$6=4+2$.

\subsection{Why the Fermat cone has a large count}

Let
\[
  Z=V(x_1^n+\cdots+x_n^n)\subset\PP^{n-1}
\]
be the smooth Fermat hypersurface, and let $X\subset\PP^n$ be its cone with vertex $q_0=[1:0:\cdots:0]$.  Every tangent hyperplane to $X$ contains $q_0$.  Therefore the top dual multidegree $\pdeg_{n-1}(\Con(X))$ vanishes.  In the next slot, a generic primal hyperplane meets each ruling line of the cone once, leaving the dual-degree count of the base:
\[
  \pdeg_{n-2}(\Con(X))=\deg Z^\vee=n(n-1)^{n-2}.
\]
The last equality is the classical dual-degree formula for a smooth degree-$n$ hypersurface in $\PP^{n-1}$.

\subsection{Why determinants admit an upper count}

At a smooth point of the determinant hypersurface where $\widehat A(x)$ has corank one, choose left and right kernel vectors $u$ and $v$.  The adjugate has rank one,
\[
  \adj\widehat A(x)=\alpha\,v u^{\mathsf T},
\]
and Jacobi's formula gives
\[
  \partial_i\det\widehat A(x)=\alpha\,u^{\mathsf T}\widehat A_i v.
\]
Thus the tangent hyperplane is encoded by the bilinear kernel pairings $u^{\mathsf T}\widehat A_i v$.  A Gauss-graph point can therefore be lifted to $(x,[u],[v])$.  After adding the generic flag conditions defining $\pdeg_{n-2}$, these lifts satisfy a square multihomogeneous system.  Theorem~\ref{thm:detlemma} proves that the desired lifts are isolated and then applies B\'ezout.

The determinant-side proof can be read as the following six-step algorithm.
\begin{enumerate}[leftmargin=2.2em,itemsep=2pt]
\item Use a generic flag to turn the desired multidegree into finitely many
smooth Gauss points $(x,[dF(x)])$.
\item At each point, use $dF(x)\ne0$ to deduce
$\rank\widehat A(x)=m-1$ and therefore unique projective left and right
kernel lines $[u]$ and $[v]$.
\item Express the tangent covector through the kernel pairings
$u^{\mathsf T}\widehat A_i v$.
\item Write all kernel and flag conditions as a square system on
$\PP^n\times\PP^{m-1}\times\PP^{m-1}$.
\item Prove locally, by a Schur-complement calculation, that every desired
lift is an isolated reduced solution of that system.
\item Apply multihomogeneous B\'ezout.  Other solutions or excess components
may exist, but they only consume nonnegative intersection multiplicity.
\end{enumerate}

\subsection{Dictionary of the proof}

\begin{center}
\small
\begin{tabularx}{\textwidth}{@{}p{0.22\textwidth}X p{0.27\textwidth}@{}}
\toprule
Object & Plain-language meaning & Role in the proof \\
\midrule
$\dcb(f)$ & Smallest determinant size that can approximate $f$ & Starting hypothesis \\
$\Gone(F)$ & Tangent-hyperplane graph of multiplicity-one components & Survives repeated factors \\
$\pdeg_{n-2}$ & One condition on the point, $n-2$ on its tangent hyperplane & Common quantity compared on both sides \\
$W_0$ & Effective flat limit of nearby Gauss graphs & Conserves the total count \\
$\Con(X)$ & Tangent data of the Fermat cone & Supplies the large lower count \\
Kernel incidence & Equations for $x$, left kernel $u$, right kernel $v$ & Supplies the determinant-side upper count \\
Dominates $C$ & A component maps densely onto the parameter curve & Keeps only geometry that occurs for moving fibres \\
Vertical component & A component supported over finitely many parameter values & Removed before taking the moving closure \\
Flat over $C$ & Fibre cycles specialize with multiplicity bookkeeping & Makes multidegrees constant \\
Reduced isolated point & Local solution ring has length one & Ensures one counted Gauss point contributes exactly one \\
\bottomrule
\end{tabularx}
\end{center}

\section{Formal setup and reusable facts}\label{sec:prelim}

\begin{readerbox}{How to use this section}
This section fixes the cycle-theoretic language used in the proof.  On a first reading, the essential points are: (i)~$\Gone(F)$ keeps exactly the components on which the differential is generically nonzero; (ii)~$\pdeg_k$ is an enumerative count; and (iii)~multidegrees of effective cycles are nonnegative.
\end{readerbox}

\begin{detailbox}{Notation at a glance}
\begin{tabularx}{\textwidth}{@{}p{0.18\textwidth}X@{}}
$F_n$ & $x_1^n+\cdots+x_n^n$, the affine target polynomial. \\
$\widetilde F_n$ & The same degree-$n$ form in $x_0,x_1,\ldots,x_n$, independent of $x_0$. \\
$X$ & The projective cone $V(\widetilde F_n)\subset\PP^n$. \\
$\Con(S)$ & Closure of pairs consisting of a smooth point of $S$ and a hyperplane tangent there. \\
$\Gone(F)$ & Sum of conormals of the multiplicity-one components of $V(F)$. \\
$\pdeg_k$ & Multidegree imposing $n-1-k$ primal and $k$ dual linear conditions. \\
$G_c,\widehat G_c$ & The affine family and its degree-$m$ homogenization. \\
$\Gamma_c$ & $\Gone(\widehat G_c)$ on a generic fibre. \\
$W$ & The closed flat family formed from the moving components of the open Gauss graph. \\
\end{tabularx}
\end{detailbox}

We work over $\CC$ throughout; $n\ge2$ is the number of affine
variables, and $F_n=x_1^n+\cdots+x_n^n$.  We write $\widetilde F_n$
for the same polynomial regarded as a degree-$n$ form in
$(x_0,x_1,\dots,x_n)$ (it does not involve $x_0$), and
\[
  X\ :=\ V(\widetilde F_n)\ \subset\ \PP^n,
\]
the projective cone over the smooth Fermat hypersurface
$Z=V(F_n)\subset\PP^{n-1}=\{x_0=0\}$ with vertex
$q_0=[1:0:\cdots:0]$.

\subsection{Conormal cycles and multidegrees}

Let $h_1,h_2$ denote the pullbacks to $\PP^n\times(\PP^n)^\vee$ of the
hyperplane classes of the two factors.  For an $(n-1)$-cycle $C$ on
$\PP^n\times(\PP^n)^\vee$ and $0\le k\le n-1$, the $k$-th
\emph{multidegree} is
\[
  \pdeg_k(C)\ :=\ \deg\big(C\cdot h_1^{\,n-1-k}h_2^{\,k}\big)\in\mathbb Z .
\]
Multidegrees are additive in $C$ and, by
Lemma~\ref{lem:kleiman} below, nonnegative on effective cycles:
$\pdeg_k$ of an irreducible $(n-1)$-fold is computed by a generic
transverse flag intersection and is the (nonnegative) number of
intersection points.

For an irreducible subvariety $S\subset\PP^n$, $\Con(S)\subset
\PP^n\times(\PP^n)^\vee$ is the conormal variety: the closure of
$\{(x,[\xi]): x\in S_{\mathrm{sm}},\ \xi|_{T_xS}=0\}$.  For a
hypersurface $S$ it has dimension $n-1$, and over a smooth point of
$S=V(g)$ ($g$ reduced) the fibre is the single covector $[dg(x)]$.

\begin{definition}[Multiplicity-one Gauss-graph cycle]\label{def:gone}
Let $F$ be a nonzero form on $\PP^n$ of degree $\ge1$, with
factorization $F=\prod_j g_j^{k_j}$ into pairwise distinct
irreducibles, $S_j=V(g_j)$.  Set
\[
  \Gone(F)\ :=\ \overline{\big\{(x,[dF(x)])\ :\ F(x)=0,\
  dF(x)\ne0\big\}}\quad\text{(closure, with the reduced structure)},
\]
regarded as the $(n-1)$-cycle given by its irreducible components with
coefficient one.
\end{definition}

\begin{lemma}[Structure of $\Gone$]\label{lem:gonestructure}
$\Gone(F)=\sum_{j:\,k_j=1}\Con(S_j)$ as cycles.  In particular
$\Gone(F)$ is effective of pure dimension $n-1$ (or the zero cycle,
when no component has multiplicity one), and if $F$ is squarefree then
$\Gone(F)$ is the full reduced conormal cycle of $V(F)$.
\end{lemma}

\begin{intuitionbox}{Why the factorization immediately identifies $\Gone(F)$}
At a generic point of one component $S_j$, all other irreducible factors are
units.  If $g_j$ occurs once, $dF$ is a nonzero unit times $dg_j$; if it
occurs at least twice, every derivative still contains a factor $g_j$ and
vanishes on $S_j$.  The proof below is exactly this calculation.
\end{intuitionbox}

\begin{proof}
On a component of multiplicity $k_j\ge2$ one has
$dF=g_j^{\,k_j-1}\big(k_j(\prod_{l\ne j}g_l^{k_l})\,dg_j+g_j(\cdots)\big)$,
which vanishes identically on $S_j$; so the open graph in
Definition~\ref{def:gone} contains no point lying \emph{only} on
multiplicity-$\ge2$ components.  Conversely, at a point $x$ lying on
$S_j$ with $k_j=1$ and on no other component, $dF(x)=
\big(\prod_{l\ne j}g_l^{k_l}(x)\big)\,dg_j(x)$, a nonzero multiple of
$dg_j(x)$ wherever $dg_j(x)\ne0$; such $x$ are dense in $S_j$ ($g_j$
irreducible, hence reduced), and there $[dF(x)]=[dg_j(x)]$ is the
conormal covector of $S_j$.  Hence the open graph is a disjoint union
of dense opens of the $\Con(S_j)$ with $k_j=1$, and its closure is
their union.  Each $\Con(S_j)$ is irreducible of dimension $n-1$.
\end{proof}

\begin{example}[A factorization sanity check]
For $F=g^2h^3k$ with distinct reduced irreducible factors, the only
multiplicity-one component is $V(k)$.  Lemma~\ref{lem:gonestructure} gives
$\Gone(F)=\Con(V(k))$, regardless of the singularities or intersections of
the three components.  This example is useful later when the homogenized
special fibre contains a high power of $x_0$ together with one Fermat factor.
\end{example}

\begin{example}[Tangency: multiplicity-one is not corank-one]
\label{ex:tangency}
The dichotomy in Lemma~\ref{lem:dichotomy} below is sharper than rank
bookkeeping, and the following example shows why the distinction
matters.  Let $m\ge2$ and
\[
  \widehat A=\begin{pmatrix}x_1&x_2\\ 0&x_1\end{pmatrix}
  \oplus x_0 I_{m-2},
  \qquad F=\det\widehat A=x_0^{\,m-2}x_1^2 .
\]
Along the component $S=\{x_1=0\}$ the matrix has corank exactly one at
a generic point, with left kernel $u=e_2$ and right kernel $v=e_1$,
yet $S$ has multiplicity two in $F$ and every kernel pairing
$u^{\mathsf T}\widehat A_iv=(\widehat A_i)_{21}$ vanishes: the
component is \emph{tangential}.  Generic corank one does not imply
multiplicity one; the differential criterion of
Lemma~\ref{lem:dichotomy} is the correct test, and $\Gone$ is defined
through it.
\end{example}

\begin{lemma}[Multiplicity dichotomy for determinants]
\label{lem:dichotomy}
In the setting of Theorem~\ref{thm:detlemma}, let $S$ be an
irreducible component of $V(F)$, $F=\det\widehat A$.  The following
are equivalent:
\begin{enumerate}[label=(\alph*),leftmargin=2.2em]
\item $S$ has multiplicity one in $F$;
\item $dF$ does not vanish identically on $S$;
\item at a generic point of $S$ the matrix $\widehat A$ has corank
exactly one, \emph{and} the kernel pairing
$w\mapsto u^{\mathsf T}\widehat A(w)\,v$ (with $u,v$ spanning the left
and right kernels) is not identically zero.
\end{enumerate}
In the symmetric case (c) reads: corank one and
$w\mapsto u^{\mathsf T}\widehat A(w)\,u\not\equiv0$.
\end{lemma}

\begin{detailbox}{What the dichotomy is testing}
Corank one is necessary because a nonzero adjugate is needed to obtain a
nonzero derivative of the determinant.  It is not sufficient: the affine
linear subspace of matrices may be tangent to the determinant hypersurface.
The kernel pairing in part~(c) tests whether the parameter directions cross
the determinant hypersurface transversely rather than merely moving inside
its tangent space.
\end{detailbox}

\begin{proof}
(a)$\Leftrightarrow$(b) is the computation in
Lemma~\ref{lem:gonestructure}.  For (b)$\Leftrightarrow$(c): since
$\widehat A$ is linear in $(x_0,\dots,x_n)$, the directional
derivative of $F$ at $x$ in direction $w$ is, by Jacobi's formula,
\begin{equation}\label{eq:jacobi}
  \partial_wF(x)
  =\tr\!\big(\adj\widehat A(x)\cdot\widehat A(w)\big).
\end{equation}
If $\widehat A$ has corank $\ge2$ at a generic point of $S$, then
$\adj\widehat A\equiv0$ on $S$ and $dF\equiv0$ on $S$ by
\eqref{eq:jacobi}.  If the corank is one at a generic point $x$, then
$\adj\widehat A(x)$ is a nonzero rank-one matrix; its columns lie in
$\ker\widehat A(x)$ and its rows in the left kernel (from
$\widehat A\cdot\adj\widehat A=\adj\widehat A\cdot\widehat A=F\cdot I
=0$ on $S$), so $\adj\widehat A(x)=\alpha(x)\,v(x)u(x)^{\mathsf T}$
with $\alpha(x)\ne0$, and \eqref{eq:jacobi} becomes
$\partial_wF(x)=\alpha(x)\,u(x)^{\mathsf T}\widehat A(w)\,v(x)$.
Thus $dF|_S\equiv0$ if and only if the pairing vanishes identically.
In the symmetric case $\adj\widehat A(x)$ is symmetric of rank one
with image equal to the kernel line, so
$\adj\widehat A(x)=\beta(x)\,u(x)u(x)^{\mathsf T}$, $\beta(x)\ne0$,
and the same computation applies.
\end{proof}

\subsection{Border complexity and closures}

\begin{lemma}[Closure agreement]\label{lem:closure}
For a constructible subset $Y$ of an affine variety over $\CC$, the
Zariski and Euclidean closures of $Y$ coincide.
\end{lemma}

\begin{proof}
Standard \cite[I.10]{Mumford99}: $Y$ contains a subset that is open
and dense, in both topologies, in the Zariski closure
$\overline{Y}^{\mathrm{Zar}}$, and Zariski-closed sets are
Euclidean-closed.
\end{proof}

\subsection{Transversality and positivity inputs}

The following two lemmas are the only intersection-theoretic inputs to
Theorem~\ref{thm:detlemma}; each is stated with the hypotheses under
which it will be invoked.

\begin{lemma}[Kleiman transversality, char.\ $0$]\label{lem:kleiman}
Let $P=\PP^{a_1}\times\cdots\times\PP^{a_r}$, let
$V_1,\dots,V_s\subset P$ be irreducible subvarieties, and let
$V_l^\circ\subseteq V_l$ be dense opens.  Let $\Phi\subset P$ be a
product of linear subspaces, of codimension equal to $\dim V_l$ for
all $l$.  For a generic translate $g\Phi$ under the transitive action
of $G=\mathrm{PGL}_{a_1+1}\times\cdots\times\mathrm{PGL}_{a_r+1}$:
each intersection $g\Phi\cap V_l$ is finite, reduced (transverse at
smooth points of $V_l$), contained in $V_l^\circ$, of cardinality
$\deg\big([V_l]\cdot[\Phi]\big)$; and the intersections for distinct
$l$ are pairwise disjoint.  In particular, multidegrees of effective
cycles are nonnegative and are computed by counting generic flag
intersections componentwise.
\end{lemma}

\begin{detailbox}{How this lemma is used}
The input cycle may have singular points and several components.  A generic
product of linear spaces misses all prescribed lower-dimensional bad sets,
meets each relevant component transversely in its smooth open part, and makes
the number of points equal to the corresponding intersection number.  No
genericity assumption is placed on the cycle itself.
\end{detailbox}

\begin{proof}
Kleiman's theorem \cite{Kleiman74} applied to each $V_l$ (transitive
group action, characteristic zero for generic reducedness of the
intersection); containment in $V_l^\circ$ and pairwise disjointness
are the same statement applied to the lower-dimensional closed sets
$V_l\smallsetminus V_l^\circ$ and $V_l\cap V_{l'}$, which a generic
translate of $\Phi$ misses for dimension reasons.
\end{proof}

\begin{lemma}[B\'ezout bound for isolated zeros]\label{lem:positivity}
Let $P=\PP^{a_1}\times\cdots\times\PP^{a_r}$, let
$E=\bigoplus_{j=1}^{D}\OO_P(d_j)$ be a direct sum of line bundles with
componentwise nonnegative multidegrees $d_j$, where $D=\dim P$, and
let $s\in H^0(P,E)$ be a section with zero scheme $Z=Z(s)$.  Then
\[
  \sum_{\substack{p\in Z\\ p\ \text{isolated in}\ Z}}
  \dim_\CC\OO_{Z,p}
  \ \le\ \int_P c_D(E)
  \ =\ \Big[\textstyle\prod_l H_l^{a_l}\Big]\
  \prod_{j=1}^{D}\Big(\textstyle\sum_l d_{j,l}H_l\Big),
\]
where $H_l$ are the hyperplane classes and the bracket extracts the coefficient in
\[
  \mathbb Z[H_1,\dots,H_r]\big/(H_1^{a_1+1},\dots,H_r^{a_r+1}).
\]
Each
isolated point contributes its local multiplicity
$\dim_\CC\OO_{Z,p}\ge1$.
\end{lemma}

\begin{warningbox}{Why excess components are allowed}
The zero scheme $Z(s)$ need not be finite.  It may contain curves or larger
components caused by singular or repeated factors.  The lemma does not try to
remove them.  It says only that the sum of the local multiplicities of the
isolated points is at most the total top-Chern number, because every excess
component contributes a nonnegative localized class.  This one-sided estimate
is precisely what the determinant proof needs.
\end{warningbox}

\begin{proof}
The section $s$ has a localized top Chern class
$\mathbf Z(s)\in A_0(Z)$ whose image in $A_0(P)$ is $c_D(E)\cap[P]$,
of degree $\int_Pc_D(E)$ \cite[\S14.1]{Fulton98}.  The class
decomposes over the connected components of $Z$.  At an isolated point
$p$ the local contribution is the intersection multiplicity
$\dim_\CC\OO_{Z,p}\ge1$ \cite[\S7.1, \S14.1]{Fulton98}.  On every
other connected component the contribution is nonnegative, because $E$
is globally generated (each $\OO_P(d_j)$ with $d_j\ge0$ is) and the
localized class of a section of a globally generated bundle is
represented by a nonnegative cycle \cite[\S12.2]{Fulton98}.  Summing
gives the bound; the coefficient formula is the Chow ring of a product
of projective spaces.  (For this statement in the language of
multihomogeneous B\'ezout numbers see also
\cite[Ch.~8]{SommeseWampler05}.)
\end{proof}

\subsection{Dual degree of a smooth hypersurface}

\begin{lemma}[Classical dual-degree formula]\label{lem:dualdeg}
Let $Z\subset\PP^{N}$ be a smooth hypersurface of degree $e\ge2$ over
$\CC$.  Then the dual variety $Z^\vee\subset(\PP^N)^\vee$ is a
hypersurface, the conormal map $\Con(Z)\to Z^\vee$ is birational
(reflexivity in characteristic zero), and
\[
  \deg Z^\vee\ =\ e(e-1)^{N-1}.
\]
Equivalently, the top multidegree of $\Con(Z)$ equals
$e(e-1)^{N-1}$.
\end{lemma}

\begin{proof}
Let $q$ be a homogeneous equation of degree $e$ for $Z$.  Smoothness means
that the partial derivatives of $q$ do not vanish simultaneously on $Z$, so
they define the Gauss morphism
\[
  \gamma:Z\longrightarrow(\PP^N)^\vee,
  \qquad x\longmapsto
  [\partial_0q(x):\cdots:\partial_Nq(x)].
\]
Each coordinate of $\gamma$ has degree $e-1$; therefore
\[
  \gamma^*\OO_{(\PP^N)^\vee}(1)=\OO_Z(e-1).
\]
A smooth hypersurface of degree at least two is not dual-defective, so the
image $Z^\vee$ is a hypersurface.  Reflexivity in characteristic zero implies
that the conormal projection to $Z^\vee$ is generically one-to-one; hence the
Gauss morphism has generic degree one onto its image
\cite[Ch.~1]{GKZ94}, \cite[\S7]{Tevelev05}.  Intersecting $Z^\vee$ with
$N-1$ generic hyperplanes and pulling those hyperplanes back by $\gamma$
gives
\[
  \deg Z^\vee
  =\int_Z c_1\big(\OO_Z(e-1)\big)^{N-1}
  =(e-1)^{N-1}\int_Z H^{N-1}.
\]
Finally, $\int_ZH^{N-1}=\deg Z=e$, because $Z$ is a degree-$e$
hypersurface.  Thus $\deg Z^\vee=e(e-1)^{N-1}$.  This is the standard
Gauss-map computation of the dual degree \cite{Piene78}.
\end{proof}

\section{Bounding the Gauss graph of an arbitrary determinant}\label{sec:detlemma}

\begin{readerbox}{What this section contributes}
The border argument later treats Theorem~\ref{thm:detlemma} as a black box.  It says that every size-$m$ affine-linear determinant has limited conormal complexity, even when the determinant is singular, reducible, or non-reduced.  Readers focused on the specialization argument may read the sketch below and then continue at Section~\ref{sec:normalform}.
\end{readerbox}

\paragraph{Proof idea before the details.}
A generic flag turns $\pdeg_{n-2}(\Gone(F))$ into a finite set of smooth Gauss-graph points.  At each such point the matrix has corank one, so the point lifts uniquely to its left and right kernel lines.  We encode these lifted points by a square multihomogeneous system in $(x,u,v)$.  A Schur-complement normal form shows that every desired lift is an isolated reduced solution; multihomogeneous B\'ezout then bounds their number.  In the symmetric model the two kernel lines coincide, which removes one projective factor and improves the constant.

\subsection{The local calculation that drives the proof}

The technically longest part of the argument is a familiar Schur-complement
calculation.  Suppose that, after constant row and column changes, a matrix
near a corank-one point has block form
\[
  \widehat A(x)=
  \begin{pmatrix}B(x)&c(x)\\ r(x)&s(x)\end{pmatrix},
  \qquad \det B(x^*)\ne0.
\]
Then $B$ remains invertible in the local ring at $x^*$, and
\[
  \det\widehat A=(\det B)g,
  \qquad g=s-rB^{-1}c.
\]
On the hypersurface germ $g=0$, the projective kernel lines are represented
by
\[
  v(x)=\binom{-B^{-1}c}{1},
  \qquad
  u(x)^{\mathsf T}=\big(-rB^{-1},1\big).
\]
Thus the kernel equations do not create a mysterious new local variety: they
solve uniquely for $u$ and $v$ as functions of $x$, while their last remaining
condition is exactly the hypersurface equation $g=0$.  Step~5 below proves
this statement scheme-theoretically and then shows that the flag equations
cut the resulting graph transversely.

For $m=2$, $B$ is a single invertible scalar function $b$, and the formulas
become
\[
  \det\begin{pmatrix}b&c\\r&s\end{pmatrix}=b\left(s-\frac{rc}{b}\right),
  \quad
  v=\left(-\frac cb,1\right),
  \quad
  u^{\mathsf T}=\left(-\frac rb,1\right).
\]
The general calculation is the same, with matrix inverses replacing scalar
division.

\subsection{Variable and equation bookkeeping in the ordinary model}

\begin{center}
\small
\begin{tabularx}{\textwidth}{@{}p{0.24\textwidth}p{0.16\textwidth}p{0.18\textwidth}X@{}}
\toprule
Block & Number & Multidegree & Purpose \\
\midrule
Point $x\in\PP^n$ & dimension $n$ & $H$ & Location on the hypersurface \\
Left kernel $[u]\in\PP^{m-1}$ & dimension $m-1$ & $U$ & Encodes the row-kernel line \\
Right kernel $[v]\in\PP^{m-1}$ & dimension $m-1$ & $V$ & Encodes the column-kernel line \\
$u^{\mathsf T}\widehat A(x)=0$ & $m$ equations & $H+U$ & Full left-kernel condition \\
$\Lambda\widehat A(x)v=0$ & $m-1$ equations & $H+V$ & Locally complete right-kernel condition \\
Dual flag equations & $n-2$ equations & $U+V$ & Conditions on the tangent hyperplane \\
Primal flag equation & $1$ equation & $H$ & One condition on $x$ \\
\bottomrule
\end{tabularx}
\end{center}
The ambient dimension is $n+2m-2$, and the total number of equations is also
$n+2m-2$.  Squareness alone does not prove finiteness; Step~5 verifies that
the particular lifts we care about are indeed reduced isolated solutions.

Throughout, $A,\widehat A,F$ are as in Theorem~\ref{thm:detlemma}, $n\ge2$, and $F=\det\widehat A\not\equiv0$.  If $\Gone(F)$ is zero there is nothing to prove.  Otherwise write $\Gone(F)=\sum_{j\in J}\Con(S_j)$ as in Lemma~\ref{lem:gonestructure}.

\subsection{Proof of Theorem~\ref{thm:detlemma}(i)}
\label{subsec:nonsymproof}

\paragraph{Step 1: count points.}
\begin{roadmapbox}{turn the multidegree into a finite list}
Each conormal component has dimension $n-1$.  Impose one generic primal
hyperplane and $n-2$ generic dual hyperplanes, for a total codimension
$n-1$.  Kleiman transversality makes the intersection finite, reduced, and
contained in the smooth open Gauss graph.  The output is exactly
$\pdeg_{n-2}(\Gone(F))$ distinct points together with the linear independence
of the corresponding flag differentials.  That differential independence is
saved for Step~5.
\end{roadmapbox}
For each $j\in J$ let $S_j^{\mathrm{good}}\subseteq S_j$ be the dense
open subset of points $x$ such that $x$ is a smooth point of
$V(F)_{\mathrm{red}}$ lying on $S_j$ only, and $dF(x)\ne0$; this is
dense because $S_j$ has multiplicity one
(Lemma~\ref{lem:dichotomy}(b)) and $F$ restricted to a neighborhood of
a generic point of $S_j$ is a unit times a reduced equation of $S_j$.
Over $S_j^{\mathrm{good}}$ the conormal variety is the graph
$x\mapsto(x,[dF(x)])$; let $\Con(S_j)^\circ$ be this graph, dense open
in $\Con(S_j)$, isomorphic to $S_j^{\mathrm{good}}$ via the first
projection.

Apply Lemma~\ref{lem:kleiman} on $\PP^n\times(\PP^n)^\vee$ to the
varieties $\{\Con(S_j)\}_{j\in J}$, the opens $\Con(S_j)^\circ$, and
the flag $\Phi=M\times(L_1\cap\cdots\cap L_{n-2})$ with $M$ a
hyperplane in $\PP^n$ and $L_t$ hyperplanes in $(\PP^n)^\vee$: for a
generic such flag, the intersection consists of exactly
$\sum_j\pdeg_{n-2}(\Con(S_j))=\pdeg_{n-2}(\Gone(F))$ distinct
\emph{count points}
\[
  (x^*,\xi^*),\qquad x^*\in S_j^{\mathrm{good}},\quad
  \xi^*=[dF(x^*)],
\]
each a transverse, reduced intersection point of the flag with the
graph $\Con(S_j)^\circ$.  Transversality at $(x^*,\xi^*)$, transported
through the isomorphism $S_j^{\mathrm{good}}\cong\Con(S_j)^\circ$,
says precisely: writing $\ell_t$ for the linear forms cutting
$L_1\cap\cdots\cap L_{n-2}$ and $M$ for the linear form cutting the
hyperplane, the $n-1$ functions
\begin{equation}\label{eq:flagfunctions}
  M(x),\quad \ell_1\big(dF(x)\big),\ \dots,\
  \ell_{n-2}\big(dF(x)\big)
\end{equation}
vanish at $x^*$ and have linearly independent differentials on
$T_{x^*}S_j$.  This is the form in which transversality will be used
in Step~5.

\paragraph{Step 2: the kernel lift.}
\begin{roadmapbox}{replace each tangent point by kernel data}
At a count point, $F(x^*)=0$ but $dF(x^*)\ne0$.  Jacobi's formula therefore
forces the adjugate to be nonzero, so the matrix has corank exactly one.
Consequently its projective left and right kernels are unique.  The rank-one
adjugate identity then expresses every coordinate of the tangent covector as
$u^{\mathsf T}\widehat A_i v$.  The output is one unique lift
$(x^*,[u^*],[v^*])$ for each Gauss point.
\end{roadmapbox}
Fix a count point $(x^*,\xi^*)$.  Since $dF(x^*)\ne0$, formula
\eqref{eq:jacobi} forces $\adj\widehat A(x^*)\ne0$, hence
$\rank\widehat A(x^*)\ge m-1$; since $F(x^*)=0$ the rank is exactly
$m-1$.  The left and right kernels are lines, spanned by $u^*,v^*$,
and as in the proof of Lemma~\ref{lem:dichotomy},
\begin{equation}\label{eq:conormalid}
  \adj\widehat A(x^*)=\alpha\,v^*(u^*)^{\mathsf T},\quad\alpha\ne0,
  \qquad
  \partial_iF(x^*)=\alpha\,(u^*)^{\mathsf T}\widehat A_i\,v^*\ \
  (0\le i\le n),
\end{equation}
so $\xi^*=\big[(u^*)^{\mathsf T}\widehat A_0v^*:\cdots:
(u^*)^{\mathsf T}\widehat A_nv^*\big]$.  Each count point therefore
lifts to the unique point
$P^*=(x^*,[u^*],[v^*])\in\PP^n\times\PP^{m-1}\times\PP^{m-1}$.

\paragraph{Step 3: choice of the reduction matrix $\Lambda$.}
\begin{roadmapbox}{remove exactly one redundant kernel equation}
The two full kernel systems contain $2m$ scalar equations, but the corank-one
kernel incidence has codimension $2m-1$.  The relation
$u^{\mathsf T}(\widehat A v)=(u^{\mathsf T}\widehat A)v$ explains why the two
blocks cannot be independent.  We retain all left equations and replace the
$m$ right equations by $m-1$ generic linear combinations.  The map
$\Lambda$ is chosen only after the finite list of count points is known, so it
can be made injective on every relevant image hyperplane simultaneously.
Step~5 proves that this reduced block defines the same local kernel line at
each lift.
\end{roadmapbox}
The incidence conditions at a lift are $u^{\mathsf T}\widehat A(x)=0$
($m$ equations) and $\widehat A(x)v=0$ ($m$ equations), but the
incidence variety over the corank-one locus has codimension $2m-1$,
not $2m$, in $\PP^n\times\PP^{m-1}\times\PP^{m-1}$: one equation is
locally redundant.  We remove the redundancy on the \emph{right} side.
Choose a linear map $\Lambda:\CC^m\to\CC^{m-1}$, generic subject to
the finitely many open conditions
\begin{equation}\label{eq:lambdacond}
  \ker\Lambda\cap\im\widehat A(x^*)=0
  \qquad\text{for every count point }x^* ,
\end{equation}
each of which excludes a proper closed subset of $\Lambda$-space
because $\ker\Lambda$ is a line and $\im\widehat A(x^*)$ is a
hyperplane of $\CC^m$.  We also require, generically, that no row of
$\Lambda\widehat A$ vanishes identically as a form (the left
annihilators $\{\lambda:\lambda^{\mathsf T}\widehat A\equiv0\}$ form a
proper linear subspace since $\widehat A\not\equiv0$).  The
scheme-theoretic content of the reduction is established by the explicit
elimination in Step~5(ii): near each lift, the left block
$u^{\mathsf T}\widehat A=0$ together with the $m-1$ equations
$\Lambda\widehat Av=0$ generates the same ideal as the left block together
with the full right block $\widehat Av=0$.  The reduced right equations
alone do not generate the full right-kernel ideal; the left block supplies
the missing hypersurface equation.  Condition \eqref{eq:lambdacond}
resurfaces in Step~5(ii) as the invertibility of a concrete
$(m-1)\times(m-1)$ matrix.

\begin{remark}[Why a row may not simply be deleted]\label{rem:norow}
Replacing the $m$ equations $\widehat A v=0$ by $m-1$ of them ---
deleting a row instead of taking $m-1$ generic combinations --- is
unsound: deleting row $k$ is the choice $\Lambda=$ (projection killing
$e_k$), for which the matrix $M_v$ of Step~5(ii) is a fixed
$(m-1)\times(m-1)$ submatrix of $\widehat A$'s column block, which can
be singular at special $x^*$ even where $\widehat A$ has corank one;
the solution set then fattens to a positive-dimensional kernel through
the lift and Step~5 fails.  The generic-$\Lambda$ reduction keeps the
solution set equal to the kernel line near every count point.  The
same caution applies on the left: the left system is kept \emph{whole}
($m$ equations), because it is what excludes solutions lying over the
locus $\det\widehat A(x)\ne0$ in Step~5(i).
\end{remark}

\paragraph{Step 4: the square multihomogeneous system.}
\begin{roadmapbox}{write one global system to which B\'ezout applies}
Combine the left-kernel equations, the reduced right-kernel equations, the
$n-2$ dual flag equations, and the one primal flag equation.  Their
multidegrees are read separately in the $x$, $u$, and $v$ variables.  The top
Chern class of the resulting direct sum of line bundles is the coefficient
$B(m,n)$.  At this stage B\'ezout gives only a total intersection budget; it
remains to show that every desired lift consumes at least one unit of that
budget as an isolated point.
\end{roadmapbox}
On $Y=\PP^n\times\PP^{m-1}_{[u]}\times\PP^{m-1}_{[v]}$, with
hyperplane classes $H,U,V$ and $\dim Y=n+2m-2$, consider the system
\begin{equation}\label{eq:system}
\begin{array}{l@{\qquad}l@{\qquad}l}
  u^{\mathsf T}\widehat A(x)=0 & m\ \text{equations} & \text{class }H+U,\\[1pt]
  \Lambda\,\widehat A(x)\,v=0 & m-1\ \text{equations} & \text{class }H+V,\\[1pt]
  \ell_t\big(u^{\mathsf T}\widehat A_0v,\dots,
        u^{\mathsf T}\widehat A_nv\big)=0,\ \ 1\le t\le n-2
        & n-2\ \text{equations} & \text{class }U+V,\\[1pt]
  M(x)=0 & 1\ \text{equation} & \text{class }H,
\end{array}
\end{equation}
where $\ell_t,M$ are the flag forms of Step~1.  The count is
$m+(m-1)+(n-2)+1=n+2m-2=\dim Y$: the system is square.  The class
assignments: each entry of $u^{\mathsf T}\widehat A$ is bilinear in
$(x,u)$ and each entry of $\Lambda\widehat Av$ is bilinear in $(x,v)$;
each polar equation is
$\ell_t(\dots)=u^{\mathsf T}\big(\sum_i\lambda_{t,i}\widehat A_i\big)v$,
bilinear in $(u,v)$ and independent of $x$ --- a genuine nonzero
section of $\OO(0,1,1)$ for generic $\ell_t$, since
$\sum_i\lambda_i\widehat A_i=0$ defines a proper linear subspace of
covectors ($\widehat A\not\equiv0$), which the generic flag avoids.
Likewise each entry of $u^{\mathsf T}\widehat A$ is a nonzero section
(a zero entry would be a zero column of $\widehat A$, forcing
$F\equiv0$), and the rows of $\Lambda\widehat A$ were arranged nonzero
in Step~3.  All multidegrees are componentwise nonnegative, so
Lemma~\ref{lem:positivity} applies with
\[
  \int_Y c_{\mathrm{top}}
  =\big[H^nU^{m-1}V^{m-1}\big]\;
  H\,(H+U)^m(H+V)^{m-1}(U+V)^{n-2}
  \ =\ B(m,n),
\]
the equality of the two displayed brackets being the symmetry
$u\leftrightarrow v$ of the extraction (the closed form is computed in
Step~6).

\paragraph{Step 5: the lifts are reduced isolated points of the
solution scheme.}
\begin{warningbox}{Technical heart of the determinant-side proof}
A square system can still have positive-dimensional components, and such
components do occur over repeated or higher-corank pieces of the determinant.
The task here is narrower: prove that each lift selected in Step~1 is locally
a single reduced point of the full system.  The proof first shows that the
kernel block is a graph over the smooth hypersurface germ, and then invokes
the transversality already obtained in Step~1 to cut that graph to one point.
\end{warningbox}
Fix a count point with lift $P^*=(x^*,[u^*],[v^*])$.

\emph{(i) No solutions off the hypersurface.}  If
$\det\widehat A(x)\ne0$ then $u^{\mathsf T}\widehat A(x)=0$ forces
$u=0$, which is not a point of $\PP^{m-1}$.  This is why the left
system is kept whole.

\emph{(ii) Local normal form: the bilinear block is
scheme-theoretically a graph over the hypersurface germ.}
After constant row and column operations
$\widehat A\mapsto P\widehat AQ$ --- which multiply $F$ by the nonzero
constant $\det P\det Q$, transform
$u\mapsto P^{-\mathsf T}u$, $v\mapsto Q^{-1}v$,
$\Lambda\mapsto\Lambda P^{-1}$ equivariantly, and leave every pairing
$u^{\mathsf T}\widehat A_iv$ and condition \eqref{eq:lambdacond}
intact --- we may assume the leading $(m-1)\times(m-1)$ block of
$\widehat A(x^*)$ is invertible (some $(m-1)$-minor is nonzero at a
rank-$(m-1)$ matrix).  Work in an affine chart of $\PP^n$ at $x^*$ and
write
\[
  \widehat A=\begin{pmatrix}B&c\\ r&s\end{pmatrix},
  \qquad \det B\in\OO_{x^*}^{\times},
\]
with $B$ of size $m-1$ and all blocks affine-linear in the chart
coordinates.  Then $u^*_m\ne0$ and $v^*_m\ne0$: a kernel vector with
vanishing last coordinate is annihilated by the invertible $B$ and
dies.  So we may work in the charts $u_m=1$, $v_m=1$, writing
$u^{\mathsf T}=(u'^{\mathsf T},1)$, $v=(v',1)^{\mathsf T}$.  Set
\[
  g\ :=\ s-rB^{-1}c\ \in\OO_{x^*},
  \qquad
  z\ :=\ v'+B^{-1}c ,
\]
the Schur complement and the recentred right-kernel coordinate, and
recall the Schur identity
$\det\widehat A=(\det B)\,g$, so that $(F)=(g)$ in $\OO_{x^*}$.  In
the local ring $\OO_{P^*}$ (variables: chart coordinates of $x$, and
$u',v'$), the ideal $I$ of the bilinear block
$\{u^{\mathsf T}\widehat A=0,\ \Lambda\widehat Av=0\}$ is computed by
elimination.

\emph{Left block.}  The first $m-1$ left equations are
$u'^{\mathsf T}B+r=0$; since $B$ is invertible over $\OO_{x^*}$ they
generate the same ideal as the entries of
$u'^{\mathsf T}+rB^{-1}$.  Modulo these, the last left equation
becomes
$u'^{\mathsf T}c+s\ \equiv\ (-rB^{-1})c+s\ =\ g$.

\emph{Right block.}  Writing $v'=z-B^{-1}c$,
\[
  \widehat A\,v
  =\begin{pmatrix}Bv'+c\\ rv'+s\end{pmatrix}
  =\begin{pmatrix}Bz\\ rz+g\end{pmatrix},
  \qquad\text{hence}\qquad
  \Lambda\widehat A\,v\ =\ M_v(x)\,z\ +\ g\,\lambda ,
\]
an exact identity, where $\Lambda=[\Lambda_{\mathrm{top}}\mid\lambda]$
with $\lambda$ the last column of $\Lambda$, and
\[
  M_v(x)\ :=\ \Lambda\,C(x),\qquad
  C(x):=\begin{pmatrix}B\\ r\end{pmatrix}
  =\text{the first }m-1\text{ columns of }\widehat A .
\]
At $x^*$, $C(x^*)$ has rank $m-1$ (its top block is $B(x^*)$) and its
column space is therefore all of $\im\widehat A(x^*)$; so
condition~\eqref{eq:lambdacond} says \emph{exactly} that $\Lambda$ is
injective on the column space of $C(x^*)$, i.e.\ that $M_v(x^*)$ is
invertible.  Hence $M_v$ is invertible over $\OO_{x^*}$, and modulo
the generators already obtained,
$\Lambda\widehat Av\equiv M_vz$ generates the same ideal as the
entries of $z$.  Altogether
\begin{equation}\label{eq:normalform}
  I\ =\ \big(\,u'^{\mathsf T}+rB^{-1},\ \ v'+B^{-1}c,\ \ g\,\big):
\end{equation}
This equality of ideals is stronger than equality of point sets.  The first
$m-1$ generators solve uniquely for the affine coordinates $u'$, the next
$m-1$ solve uniquely for $v'$, and the final generator is the single equation
$g=0$ for the determinant hypersurface.  In particular there is no hidden
nilpotent thickening in the kernel directions: the local quotient ring is
canonically the local ring of $V(g)$.

Consequently, the bilinear block is, scheme-theoretically, the graph
\[
  x\ \longmapsto\ \big(x,\,[u(x)],\,[v(x)]\big),
  \qquad
  u(x)^{\mathsf T}=(-rB^{-1},\,1),\quad
  v(x)=(-B^{-1}c,\,1)^{\mathsf T},
\]
over the germ $V(g)$.  Since $(F)=(g)$ and $dF(x^*)\ne0$, that germ is
the smooth reduced hypersurface germ $(S_j,x^*)$, and the graph is a
smooth germ of dimension $n-1$ through $P^*$.

\emph{(iii) The conormal identity along the graph, and
transversality.}  Along $(S_j,x^*)$ the explicit sections satisfy
$\widehat A(x)v(x)=(Bz,rz+g)^{\mathsf T}\big|_{z=0,\,g=0}=0$ and
$u(x)^{\mathsf T}\widehat A(x)=(0,\,g)\big|_{g=0}=0$: they span the
kernels.  At each such point $\widehat A$ has rank $m-1$, so
$\adj\widehat A(x)=\alpha(x)\,v(x)u(x)^{\mathsf T}$; evaluating the
$(m,m)$ entries pins the scalar, since
$(\adj\widehat A)_{mm}=\det B$ and $(vu^{\mathsf T})_{mm}=1$, giving
$\alpha=\det B$, a unit.  Jacobi's formula \eqref{eq:jacobi} then
yields, on $(S_j,x^*)$ and for every $i$,
\[
  \partial_iF(x)\ =\ (\det B)(x)\cdot
  u(x)^{\mathsf T}\widehat A_i\,v(x).
\]
Consequently each polar equation of \eqref{eq:system} restricts on the
graph to $(\det B)^{-1}\,\ell_t(dF(x))$, a unit multiple of the
corresponding flag function \eqref{eq:flagfunctions}; both vanish at
$x^*$, so their differentials at $x^*$ agree up to the unit value
$(\det B)(x^*)^{-1}$.  The remaining equation $M(x)$ restricts to
itself.  By Step~1, the $n-1$ flag functions
\eqref{eq:flagfunctions} have linearly independent differentials on
$T_{x^*}S_j$.  Therefore the full system \eqref{eq:system} cuts, near
$P^*$, the smooth $(n-1)$-dimensional germ \eqref{eq:normalform} by
$n-1$ functions with independent differentials: a single reduced
point.  $P^*$ is a reduced isolated point of the solution scheme, and
distinct count points have distinct $x^*$, hence distinct lifts.

The role of the count-point condition is visible in the normal form:
the last left equation contributes the hypersurface equation $g$, and
the polar equations contribute transverse slices precisely because
$dF(x^*)\ne0$.  On a tangential component (Example~\ref{ex:tangency})
every lifted polar equation would restrict to zero along the graph and
no isolated point would form; the dichotomy of
Lemma~\ref{lem:dichotomy} is what the construction runs on.

\paragraph{Step 6: conclusion and the closed form.}
\begin{roadmapbox}{charge the desired points to the B\'ezout budget}
Step~5 gives one reduced isolated solution for every multidegree count point.
Lemma~\ref{lem:positivity} bounds the sum of multiplicities of all isolated
solutions by the top-Chern coefficient, even if excess components are also
present.  The final task is purely combinatorial: expand the coefficient and
match the exponents of $H,U,V$.
\end{roadmapbox}
By Lemma~\ref{lem:positivity}, the number of isolated solutions of
\eqref{eq:system}, counted with multiplicity $\ge1$, is at most
$B(m,n)$; by Step~5 the $\pdeg_{n-2}(\Gone(F))$ lifts are among them,
each reduced.  Hence $\pdeg_{n-2}(\Gone(F))\le B(m,n)$.  For the
closed form, extract the bracket.  Choose $x^i u^{m-i}$ from
$(x+u)^m$, choose $x^j v^{m-1-j}$ from $(x+v)^{m-1}$, and choose
$u^k v^{n-2-k}$ from $(u+v)^{n-2}$.  Including the leading factor
$x$, the target monomial $x^nu^{m-1}v^{m-1}$ imposes
\[
\begin{array}{c|c|c}
\text{variable} & \text{exponent equation} & \text{consequence}\\ \hline
x & 1+i+j=n & j=n-1-i\\
u & (m-i)+k=m-1 & k=i-1\\
v & (m-1-j)+(n-2-k)=m-1 & \text{automatic after the first two}
\end{array}
\]
Thus the free index is $i$, with $1\le i\le n-1$, and
\[
  B(m,n)=\sum_{i=1}^{n-1}\binom{m}{i}\binom{m-1}{n-1-i}
  \binom{n-2}{i-1},
\]
the limits enforced by $k\ge0$ and $j\ge0$.
\qed

\begin{checklistbox}{ordinary determinant bound}
The proof did not assume that $V(F)$ is smooth, irreducible, reduced, or
generically corank one on every component.  Generic flags were used only to
select count points on the multiplicity-one components.  Higher-corank and
tangential components may create excess solution sets, but
Lemma~\ref{lem:positivity} allows them.  The generic reduction matrix
$\Lambda$ was chosen after the finite count set, and Step~5 verified its local
scheme-theoretic validity at every desired lift.
\end{checklistbox}

\subsection{Proof of Theorem~\ref{thm:detlemma}(ii)}
\label{subsec:symproof}

The symmetric proof is given in full; it runs on
$Y=\PP^n\times\PP^{m-1}_{[u]}$, of dimension $n+m-1$.

\begin{center}
\small
\begin{tabularx}{\textwidth}{@{}p{0.25\textwidth}XX@{}}
\toprule
& Ordinary model & Symmetric model \\
\midrule
Kernel data & Separate left and right lines $[u],[v]$ & One kernel line $[u]$ \\
Ambient space & $\PP^n\times\PP^{m-1}\times\PP^{m-1}$ & $\PP^n\times\PP^{m-1}$ \\
Tangent coordinates & $u^{\mathsf T}\widehat A_i v$ & $u^{\mathsf T}\widehat A_i u$ \\
Dual flag class & $U+V$ & $2U$ \\
Kernel reduction & One generic $\Lambda$ is needed & No reduction is needed \\
Final coefficient & $B(m,n)$ & $2^{n-2}\binom{m}{n-1}$ \\
\bottomrule
\end{tabularx}
\end{center}

\paragraph{Step 1: count points.}
\begin{detailbox}{What is unchanged}
The multidegree and the multiplicity-one Gauss graph depend only on the
polynomial $F$, not on whether its determinantal representation is symmetric.
The same generic flag therefore produces the same finite set of transverse
count points.
\end{detailbox}
Step~1 of Section~\ref{subsec:nonsymproof} concerns only the cycle
$\Gone(F)$ and not the representation; it applies verbatim and
furnishes $\pdeg_{n-2}(\Gone(F))$ count points $(x^*,\xi^*)$ with
$x^*\in S_j^{\mathrm{good}}$, $\xi^*=[dF(x^*)]$, together with the
transversality statement for the flag functions
\eqref{eq:flagfunctions} on $T_{x^*}S_j$.

\paragraph{Step 2: the kernel lift.}
\begin{detailbox}{Why one kernel line suffices}
For a symmetric corank-one matrix, the left and right kernels coincide.  Its
rank-one adjugate is symmetric and hence a scalar multiple of $uu^{\mathsf T}$.
This replaces bilinear tangent coordinates by quadratic coordinates and
removes one projective factor from the incidence.
\end{detailbox}
At a count point, $dF(x^*)\ne0$ and \eqref{eq:jacobi} force
$\adj\widehat A(x^*)\ne0$, hence $\rank\widehat A(x^*)=m-1$ as
before.  Since $\widehat A(x^*)$ is symmetric, its adjugate is a
nonzero \emph{symmetric} rank-one matrix whose image is the kernel
line, so
\[
  \adj\widehat A(x^*)=\beta\,u^*(u^*)^{\mathsf T},\quad\beta\ne0,
  \qquad
  \partial_iF(x^*)=\beta\,(u^*)^{\mathsf T}\widehat A_i\,u^*\ \
  (0\le i\le n),
\]
and the count point lifts to the unique
$P^*=(x^*,[u^*])\in\PP^n\times\PP^{m-1}$.

\paragraph{Step 3: squareness without a reduction.}
\begin{detailbox}{Equation count}
The $m$ kernel equations, the $n-2$ dual flag equations, and the one primal
flag equation total $m+n-1$, exactly the dimension of
$\PP^n\times\PP^{m-1}$.  The last kernel equation supplies the local
hypersurface equation, so deleting one would be unsound.
\end{detailbox}
The kernel condition $\widehat A(x)u=0$ is $m$ equations, and the
system below has $m+(n-2)+1=n+m-1=\dim Y$ equations: square as it
stands.  No analogue of the $\Lambda$-reduction is needed --- the
elimination in Step~5(ii) shows the $m$ kernel equations cut exactly
codimension $m$ near the lift, the last of them supplying the
hypersurface equation --- and none is sound to add: removing an
equation enlarges the local solution set to the kernel of an
$(m-1)\times m$ submatrix, which can be positive-dimensional through
the lift, exactly as in Remark~\ref{rem:norow}.

\paragraph{Step 4: the system and its classes.}
\begin{detailbox}{Where the factor $2^{n-2}$ comes from}
Each tangent coordinate is quadratic in the single kernel vector $u$.
Therefore each of the $n-2$ dual flag equations has class $2U$, contributing
one factor of $2$ to the top-Chern coefficient.
\end{detailbox}
The equations are
\begin{equation}\label{eq:symsystem}
\begin{array}{l@{\qquad}l@{\qquad}l}
  \widehat A(x)\,u=0 & m\ \text{equations} & \text{class }H+U,\\[1pt]
  \ell_t\big(u^{\mathsf T}\widehat A_0u,\dots,
        u^{\mathsf T}\widehat A_nu\big)=0,\ \ 1\le t\le n-2
        & n-2\ \text{equations} & \text{class }2U,\\[1pt]
  M(x)=0 & 1\ \text{equation} & \text{class }H,
\end{array}
\end{equation}
with $\ell_t,M$ the flag forms of Step~1.  Each kernel entry is
bilinear in $(x,u)$; each polar equation is
$u^{\mathsf T}(\sum_i\lambda_{t,i}\widehat A_i)u$, a quadratic form in
$u$ alone.  Nonzero sections: a vanishing row of $\widehat A$ forces
$F\equiv0$; and a flag form with
$u^{\mathsf T}(\sum_i\lambda_i\widehat A_i)u\equiv0$ forces
$\sum_i\lambda_i\widehat A_i=0$ --- a symmetric matrix over $\CC$
whose quadratic form vanishes identically is zero in characteristic
zero --- which is a proper linear condition on $\lambda$ avoided by
the generic flag.  All multidegrees are componentwise nonnegative.

\paragraph{Step 5: the lifts are reduced isolated points.}
\begin{detailbox}{Symmetric local normal form}
A congruence change makes the leading block invertible.  The first $m-1$
kernel equations solve for $u'$, and the last becomes the Schur complement
$g$.  The kernel incidence is therefore a graph over the smooth hypersurface
germ, exactly as in the ordinary case but with no separate right-kernel
variables.
\end{detailbox}
Fix a count point with lift $P^*=(x^*,[u^*])$.

\emph{(i) No solutions off the hypersurface.}  Where
$\det\widehat A(x)\ne0$, $\widehat A(x)u=0$ forces $u=0$, not a point
of $\PP^{m-1}$.

\emph{(ii) Local normal form.}  After a constant congruence
$\widehat A\mapsto P^{\mathsf T}\widehat AP$ --- which multiplies $F$
by $(\det P)^2\ne0$, transforms $u\mapsto P^{-1}u$, preserves symmetry
and every quadratic pairing
$u^{\mathsf T}\widehat A_iu$ --- assume the leading
$(m-1)\times(m-1)$ block is invertible at $x^*$, and write, in an
affine chart at $x^*$,
\[
  \widehat A=\begin{pmatrix}B&c\\ c^{\mathsf T}&s\end{pmatrix},
  \qquad \det B\in\OO_{x^*}^{\times},
  \qquad B=B^{\mathsf T}.
\]
Then $u^*_m\ne0$ (a kernel vector with last coordinate zero dies
against $B$); work in the chart $u_m=1$, $u=(u',1)^{\mathsf T}$.  Set
$g:=s-c^{\mathsf T}B^{-1}c$ and $z:=u'+B^{-1}c$, so
\[
  \widehat A\,u
  =\begin{pmatrix}Bu'+c\\ c^{\mathsf T}u'+s\end{pmatrix}
  =\begin{pmatrix}Bz\\ c^{\mathsf T}z+g\end{pmatrix},
  \qquad
  \det\widehat A=(\det B)\,g .
\]
Since $B$ is invertible over $\OO_{x^*}$, the first $m-1$ entries
generate the ideal of the entries of $z$, and modulo these the last
entry becomes $g$.  Hence the ideal of the kernel block in
$\OO_{P^*}$ is exactly
\[
  \big(\,u'+B^{-1}c,\ \ g\,\big):
\]
the kernel block is, scheme-theoretically, the graph
$u=\nu(x):=(-B^{-1}c,\,1)^{\mathsf T}$ over the germ $V(g)$, which
equals the smooth reduced germ $(S_j,x^*)$ because $(F)=(g)$ and
$dF(x^*)\ne0$; a smooth germ of dimension $n-1$ through $P^*$.

\emph{(iii) Conormal identity and transversality.}  Along
$(S_j,x^*)$ the section $\nu(x)$ spans the kernel
($\widehat A\nu=(Bz,c^{\mathsf T}z+g)^{\mathsf T}|_{z=0,g=0}=0$), the
rank is $m-1$, and the symmetric adjugate is a scalar multiple of
$\nu\nu^{\mathsf T}$; the $(m,m)$ cofactor is $\det B$ while
$(\nu\nu^{\mathsf T})_{mm}=1$, so
$\adj\widehat A=(\det B)\,\nu\nu^{\mathsf T}$ on $(S_j,x^*)$, and
Jacobi's formula \eqref{eq:jacobi} gives
\[
  \partial_iF(x)\ =\ (\det B)(x)\cdot
  \nu(x)^{\mathsf T}\widehat A_i\,\nu(x)
  \qquad\text{on }(S_j,x^*).
\]
Each polar equation of \eqref{eq:symsystem} therefore restricts on the
graph to $(\det B)^{-1}\ell_t(dF(x))$, a unit multiple of the flag
function, vanishing at $x^*$; $M(x)$ restricts to itself; and by
Step~1 the $n-1$ flag functions \eqref{eq:flagfunctions} have
independent differentials on $T_{x^*}S_j$.  The full system
\eqref{eq:symsystem} thus cuts, near $P^*$, the smooth
$(n-1)$-dimensional graph by $n-1$ functions with independent
differentials: a single reduced point.  Distinct count points give
distinct lifts.

\paragraph{Step 6: conclusion and the coefficient.}
\begin{detailbox}{Final extraction}
The top-Chern class is $H(H+U)^m(2U)^{n-2}$.  To obtain
$H^nU^{m-1}$, the factor $(H+U)^m$ must contribute
$H^{n-1}U^{m-n+1}$, whose coefficient is $\binom{m}{n-1}$.
\end{detailbox}
Lemma~\ref{lem:positivity} with
$E=\OO(1,1)^{\oplus m}\oplus\OO(1,0)\oplus\OO(0,2)^{\oplus(n-2)}$
bounds the isolated solutions of \eqref{eq:symsystem} by
\[
  \big[H^nU^{m-1}\big]\,H\,(H+U)^m(2U)^{n-2}
  \ =\ 2^{n-2}\,\big[H^{n-1}U^{m-n+1}\big](H+U)^m
  \ =\ 2^{n-2}\binom{m}{n-1},
\]
and by Step~5 the $\pdeg_{n-2}(\Gone(F))$ lifts are among them, each
reduced.
\qed

\begin{checklistbox}{why the symmetric constant is better}
The geometric lower count is unchanged.  The improvement comes entirely from
the determinant-side incidence: symmetry identifies the two kernel lines,
removes one projective factor and $m-1$ kernel equations, and changes each
bilinear dual condition into a quadratic condition on one variable.  The
resulting coefficient has exponential base $2em/(n-1)$ rather than
$4em/(n-1)$.
\end{checklistbox}

\begin{remark}[What multiplicity-$\ge2$ components do, and why no
codimension hypothesis is needed]\label{rem:excess}
Components of $V(F)$ of multiplicity $\ge2$ are of two kinds by
Lemma~\ref{lem:dichotomy}: \emph{corank-$\ge2$} components, over which
the kernel fibres of the incidence are positive-dimensional, and
\emph{tangential} corank-one components (Example~\ref{ex:tangency}),
over which the kernel incidence is still a graph but every lifted
polar equation vanishes identically.  Both populate the solution
schemes of \eqref{eq:system} and \eqref{eq:symsystem} with
positive-dimensional excess components.  Neither is bounded, located,
nor assumed away: Lemma~\ref{lem:positivity} charges excess components
a nonnegative amount and the isolated lifts at most the stated
coefficient, which is all that is used.  In particular the
corank-$\ge2$ locus may be a divisor in $V(F)$ --- e.g.\
$\widehat A=\mathrm{diag}(x_1,\dots,x_n,x_0,\dots,x_0)$ has corank
$\ge2$ along every $\{x_i=x_j=0\}$ --- and no hypothesis excludes
this.
\end{remark}

\begin{remark}[Recovery of the exact companions]\label{rem:recover}
Let $f$ be homogeneous of degree $d\ge2$ in $N\ge3$ variables with
$V(f)\subset\PP^{N-1}$ \emph{smooth}, and $f=\det M$ with $M$ of size
$m$.  Then $F:=\det\widehat M=x_0^{\,m-d}f$ on $\PP^N$, the cone
$V(f)\subset\PP^N$ is a multiplicity-one component of $V(F)$ ($f$ is
squarefree and $x_0\nmid f$), so
$\Gone(F)\supseteq\Con(\mathrm{cone}\ V(f))$ as cycles, and
Theorem~\ref{thm:detlemma}(i) with $n=N$ plus the cone-shift identity
(Proposition~\ref{prop:coneshift}, with Lemma~\ref{lem:dualdeg})
gives
\[
  d(d-1)^{N-2}=\deg V(f)^\vee
  =\pdeg_{N-2}\big(\Con(\mathrm{cone}\ V(f))\big)
  \le\pdeg_{N-2}\big(\Gone(F)\big)\le B(m,N),
\]
which is the main inequality of \cite{SheshadriDC} (the coefficient
polynomials agree under the reindexing $a=i-1$); the symmetric variant
recovers that of \cite{SheshadriSDC} the same way.  The hypotheses
here are weaker --- smoothness of $V(\det\widehat A)$ is nowhere
assumed --- which is exactly the strengthening the border transfer
requires, since the generic members of a degenerating family come with
no smoothness guarantee.
\end{remark}

\begin{remark}[The reduced conormal]\label{rem:gammared}
For squarefree $F$, $\Gone(F)$ is the full reduced conormal cycle and
Theorem~\ref{thm:detlemma} bounds it.  Whether
$\pdeg_{n-2}$ of the reduced conormal cycle of $V(\det\widehat
A)_{\mathrm{red}}$ obeys the same bound \emph{without} squarefreeness
is open; the natural attack (perturb $\widehat A$ to make the
determinant squarefree and specialize back) is itself a degeneration
argument of the kind in Sections~\ref{sec:family}--\ref{sec:capture},
and we record the question rather than entangle the two halves of this
paper.  Nothing below needs it: the border chain is built on $\Gone$
consistently on both sides.
\end{remark}

\section{From border membership to a one-parameter family}\label{sec:normalform}

\begin{readerbox}{Main message}
Membership in the closure $\overline{D_m}$ can be tested along a curve.  We therefore replace the abstract border representation by a regular polynomial family $G_c$ with $G_0=F_n$ and with each generic fibre admitting some exact size-$m$ determinantal representation.  No coherent matrix family is needed.
\end{readerbox}

\begin{detailbox}{Why a curve is enough}
The set $D_m$ sits in the finite-dimensional affine space of polynomials of
degree at most $m$.  A point in its closure can be approached inside the
closure along an algebraic curve.  Working over a curve is the simplest
setting in which specialization has a single parameter, local rings at closed
points are discrete valuation rings after normalization, and fibre cycles can
be compared directly.  The construction does \emph{not} claim that one can
choose the matrices $A_c$ algebraically in $c$; only the polynomials $G_c$
form a coherent family.
\end{detailbox}

From here through Section~\ref{sec:proofs}, fix $n\ge3$ and suppose
$\dcb(F_n)\le m$ (the symmetric case is identical with
$D_m^{\mathrm{sym}}$ and is taken up in Section~\ref{sec:proofs}).

\begin{lemma}[One-parameter normal form, fiberwise representations]
\label{lem:normalform}
There exist a smooth irreducible affine curve $C$ over $\CC$, a closed
point $0\in C$, and a family
$G\in\OO(C)[x_1,\dots,x_n]_{\le m}$ with regular coefficient
functions, such that:
\begin{enumerate}[label=(\alph*),leftmargin=2.2em]
\item $G_0=F_n$;
\item there is a dense open $C^\circ\subseteq C\smallsetminus\{0\}$
such that for every closed point $c\in C^\circ$ there exists
$A_c\in\Mat_m(\CC[x]_{\le1})$ with $\det A_c=G_c$ \emph{exactly}.
\end{enumerate}
\end{lemma}

\begin{roadmapbox}{from a closure point to fibrewise exact representations}
There are two cases.  If $F_n$ already has an exact size-$m$
representation, use a constant family.  Otherwise, curve selection gives an
irreducible curve through $F_n$ whose general point lies in $D_m$.  Normalize
the curve so that the parameter space is smooth.  Pulling back the coordinate
functions of polynomial space gives the regular coefficients of $G$.
Membership of a fibre in $D_m$ then supplies an exact representation for that
fibre by definition.
\end{roadmapbox}

\begin{proof}
If $F_n\in D_m$, take $C=\AAff^1$, $G$ the constant family, and
$A_c$ a fixed representation.  Otherwise
$F_n\in\overline{D_m}\smallsetminus D_m$.  By algebraic curve
selection through a constructible set
\cite[I.10]{Mumford99} (in the border-complexity setting,
\cite[\S20.6]{Buergisser00}), since $D_m$ is irreducible constructible
in the affine space $\CC[x]_{\le m}$ and $F_n\in\overline{D_m}$, there
is an irreducible affine curve $C'\subseteq\overline{D_m}$ with
$F_n\in C'$ and $C'\cap D_m$ dense in $C'$.  Let $\nu:C\to C'$ be the
normalization restricted to an affine neighborhood of a point
$0\in\nu^{-1}(F_n)$; $C$ is a smooth irreducible affine curve.  The
composite $C\to C'\hookrightarrow\CC[x]_{\le m}$ is a morphism to
affine space, i.e.\ a tuple of regular functions on $C$: these are the
coefficients of $G$, and $G_0=F_n$.  The set $\nu^{-1}(C'\cap D_m)$ is
dense constructible in the curve $C$, hence contains a dense open,
which we shrink to a $C^\circ$ avoiding $0$.  For $c\in C^\circ$, the
statement $G_c\in D_m$ \emph{is}, by definition of $D_m$, the
existence of an exact affine-linear representation $A_c$ over $\CC$.
\end{proof}

\begin{remark}[What the normal form does \emph{not} provide]
\label{rem:fiberwise}
No matrix family over $\OO(C^\circ)$, $\CC[t,t^{-1}]$, or
$\CC((t))$ is constructed, because nothing downstream consumes one:
Sections~\ref{sec:homog}--\ref{sec:proofs} use only the regular
polynomial family $G$ and the fiberwise representations $A_c$, with no
compatibility imposed across fibres.  This removes any
field-of-definition or Laurent-truncation step from the normal form.
The single base change in the proof (the normalization $\nu$) alters
neither the hypothesis --- $F_n\in\overline{D_m}$ concerns the fixed
set $D_m$ --- nor the parameter $m$.
\end{remark}

\paragraph{Conventions for the family.}
Several later statements hold after deleting finitely many exceptional
parameter values.  Instead of introducing a new open set each time, we keep a
single finite exclusion set $\Sigma$ and enlarge it when necessary.  Removing
finitely many points never destroys the existence of parameters approaching
$0$ along a punctured analytic disk.

Fix once and for all a finite subset $\Sigma\subset C$ of
\emph{excluded points}, enlarged finitely many times in
Lemmas~\ref{lem:nonzero}, \ref{lem:family},
and~\ref{lem:genericfibers} below, and set
$U:=C^\circ\smallsetminus\Sigma$, always cofinite in $C^\circ$ and in
particular infinite.  Statements ``for $c\in U$'' hold for every
closed point of $U$.  Analytic limits ``$c\to0$'' are taken along an
analytic disk in $C^{\mathrm{an}}$ centered at $0$; since $\Sigma$ is
finite, every punctured analytic neighborhood of $0$ meets $U$ in a
set accumulating at $0$.

\begin{lemma}\label{lem:nonzero}
After enlarging $\Sigma$: $G_c\ne0$ for every $c\in U$.
\end{lemma}

\begin{proof}
The coefficients of $G$ are regular on $C$ and not all identically
zero ($G_0=F_n\ne0$); a nonzero regular function on a curve vanishes
at finitely many points.  (Note $0\in D_m$, so membership in $D_m$
alone would not exclude $G_c=0$; this lemma does.)
\end{proof}

\section{Homogenization: why the degree must be $m$}\label{sec:homog}

\begin{readerbox}{Main message}
A size-$m$ affine determinant homogenizes to a degree-$m$ determinant.  Homogenizing the target family at any smaller degree would break the identity with the homogenized matrix.  At the special fibre the only extra factor is $x_0^{m-n}$, and the Fermat cone still occurs with multiplicity one.
\end{readerbox}

\begin{example}[What padding looks like when $m=n+2$]
At the special parameter,
\[
  \widehat G_0=x_0^2\widetilde F_n.
\]
The hyperplane $x_0=0$ occurs with multiplicity two and is invisible to the
open differential graph at its generic points.  The Fermat cone
$V(\widetilde F_n)$ occurs with multiplicity one and is visible.  The
factorization therefore creates non-reduced padding without weakening the
coefficient of the component used in the lower bound.
\end{example}

\begin{lemma}\label{lem:mgen}
$m\ge n$.
\end{lemma}

\begin{detailbox}{Degree is already a closed obstruction}
Every polynomial in $D_m$ has degree at most $m$, and the subspace of
polynomials of degree at most $m$ is closed.  Therefore a degree-$n$ target
cannot lie in $\overline{D_m}$ when $m<n$.  The coefficient-function proof
below spells this out inside the chosen curve family.
\end{detailbox}

\begin{proof}
If $m<n$, then $\deg_xG_c\le m<n$ for every $c\in C$ --- the
coefficient function of the monomial $x_1^n$ is identically zero on
$C$ --- contradicting $G_0=F_n$.
\end{proof}

\begin{definition}\label{def:homog}
$\widehat G:=x_0^{\,m}\,G(x/x_0)\in\OO(C)[x_0,\dots,x_n]_m$, the
\emph{degree-$m$ homogenization} of the family.  Its coefficients are
a reindexing of the coefficients of $G$, hence regular on $C$, and in
the chart $x_0=1$ one has $\widehat G_c=G_c$.
\end{definition}

\begin{lemma}[Fiberwise determinantal identity]\label{lem:fiberdet}
For $c\in C^\circ$:\ \ $\widehat G_c=\det\widehat A_c$, where
$\widehat A_c(x_0,x):=x_0A_c(x/x_0)$ is an $m\times m$ matrix of
linear forms on $\PP^n$ (symmetric if $A_c$ is).
\end{lemma}

\begin{proof}
$\det\widehat A_c=\det\big(x_0A_c(x/x_0)\big)
=x_0^m\det A_c(x/x_0)=x_0^mG_c(x/x_0)=\widehat G_c$.
\end{proof}

\begin{remark}[The homogenization degree]\label{rem:degreem}
The determinant of the homogenized matrix is the \emph{degree-$m$}
homogenization of $G_c$.  If $d=\deg_xG_c<m$, the degree-$d$
homogenization differs by the factor $x_0^{m-d}$:
\[
  \det\widehat A_c=x_0^{m-d}\widehat G_c^{(d)}.
\]
The extra hyperplane component has conormal
$V(x_0)\times\{[1:0:\cdots:0]\}$, whose multidegrees vanish in every
slot $k\ge1$.  Thus either convention gives the same
$\pdeg_{n-2}$ for $n\ge3$, but degree-$m$ homogenization keeps the
determinantal identity exact and is used throughout.
\end{remark}

\begin{lemma}[The special fibre]\label{lem:specialfiber}
$\widehat G_0=x_0^{\,m-n}\,\widetilde F_n$, and the Fermat cone
$X=V(\widetilde F_n)$ is a \emph{multiplicity-one} component of the
divisor $V(\widehat G_0)$.
\end{lemma}

\begin{roadmapbox}{separate the padding factor from the target factor}
First compute the special homogeneous form.  Then check three facts:
$\widetilde F_n$ is irreducible, $x_0$ does not divide it, and its exponent in
the factorization is one.  These facts locate every extra multiplicity on the
hyperplane at infinity rather than on the Fermat cone.
\end{roadmapbox}

\begin{proof}
$\widehat G_0=x_0^mF_n(x/x_0)=x_0^{\,m-n}\widetilde F_n$ since $F_n$
is homogeneous of degree $n$ and $m\ge n$
(Lemma~\ref{lem:mgen}).  The form $\widetilde F_n$ is irreducible: the
hypersurface $Z=V(F_n)\subset\PP^{n-1}$ is smooth
(Section~\ref{sec:coneshift}) and connected (a hypersurface in
$\PP^{n-1}$, $n-1\ge2$), hence irreducible, so its squarefree defining
form is irreducible, and the cone has the same defining form.  Also
$x_0\nmid\widetilde F_n$.  Hence the irreducible factorization of
$\widehat G_0$ is $x_0^{\,m-n}\cdot\widetilde F_n$, and the Fermat
component carries multiplicity exactly one: all padding multiplicity
sits on the hyperplane at infinity, none on the Fermat component,
whatever $m-n\ge0$ is.
\end{proof}

\section{Conservation in a family of Gauss graphs}
\label{sec:family}

\begin{readerbox}{The specialization principle}
Take the components of the open Gauss graph that dominate the parameter curve and close them up.  The resulting family $W\to C$ is flat.  Therefore its fibre cycle is effective and its multidegrees are constant.  The special fibre may acquire extra components, but they carry nonnegative coefficients and can only add nonnegative contributions to the special-fibre total.
\end{readerbox}

The construction of $W$ has four operations, each serving a different purpose:
\begin{enumerate}[leftmargin=2.2em,itemsep=2pt]
\item form the open graph where $\widehat G_c=0$ and
$d\widehat G_c\ne0$;
\item decompose that graph into irreducible components;
\item discard components whose image in the parameter curve is finite
(``vertical'' components); and
\item close the remaining components in the projective ambient space over the
whole curve, including the special parameter $0$.
\end{enumerate}

\begin{figure}[h]
\centering
\begin{tikzpicture}[
  >=Latex,node distance=6mm and 7mm,
  every node/.style={font=\footnotesize,align=center},
  bx/.style={draw,rounded corners,inner sep=4pt,text width=2.55cm}]
\node[bx,fill=ReaderBlue] (open) {Open smooth graph $\Gfam^\circ$};
\node[bx,right=of open] (decomp) {Irreducible components $E_i$};
\node[bx,right=of decomp,fill=ReaderGold] (discard) {Discard components with finite image in $C$};
\node[bx,right=of discard,fill=ReaderGreen] (close) {Close moving components to obtain $W$};
\draw[->] (open)--(decomp);
\draw[->] (decomp)--(discard);
\draw[->] (discard)--(close);
\end{tikzpicture}
\caption{Construction of the moving Gauss-graph family.  Only components that occur for a dense set of parameters are retained.}
\Description{Four boxes show the open Gauss graph, its irreducible decomposition, removal of vertical components, and closure to the family W.}
\label{fig:family-construction}
\end{figure}

For $c\in U$ set $\Gamma_c:=\Gone(\widehat G_c)$, the cycle of
Definition~\ref{def:gone} for the fibre.  Define the locally closed
algebraic set
\[
  \Gfam^\circ:=\big\{(x,\xi,c)\in
  \PP^n\times(\PP^n)^\vee\times C^\circ\ :\
  \widehat G_c(x)=0,\ d\widehat G_c(x)\ne0,\
  \xi=[d\widehat G_c(x)]\big\},
\]
with the reduced structure: it is the graph of the morphism
$(x,c)\mapsto[d\widehat G_c(x)]$ over the open subset
$\{d\widehat G\ne0\}$ of the closed set
$\{\widehat G=0\}\subset\PP^n\times C^\circ$, the coefficients of
$\widehat G$ and its partials being regular in $c$.

\begin{lemma}[The open graph is smooth over the curve]
\label{lem:smoothgraph}
The projection $\Gfam^\circ\to C^\circ$ is a smooth morphism of
relative dimension $n-1$.  In particular $\Gfam^\circ$ is reduced, and
its fibres are smooth of pure dimension $n-1$.
\end{lemma}

\begin{roadmapbox}{why the open graph is smooth}
Forget the dual coordinate for a moment.  The open zero set
$Z^\circ=\{\widehat G=0,\ d_x\widehat G\ne0\}$ is a relative
hypersurface whose derivative in the projective variables is nonzero.  The
Jacobian criterion makes $Z^\circ$ smooth, and the projection to the curve has
surjective differential.  Adding the tangent covector is merely taking the
graph of a morphism, which preserves smoothness.
\end{roadmapbox}

\begin{proof}
$\Gfam^\circ$ is the image of
\[
  Z^\circ:=\big\{(x,c)\in\PP^n\times C^\circ\ :\
  \widehat G_c(x)=0,\ d_x\widehat G_c(x)\ne0\big\}
\]
under the graph embedding
$(x,c)\mapsto(x,[d_x\widehat G_c(x)],c)$, an isomorphism onto
$\Gfam^\circ$ over $C^\circ$ (it is the graph of a morphism
$Z^\circ\to(\PP^n)^\vee$); so it suffices to prove
$Z^\circ\to C^\circ$ smooth.  Work near a point
$(x^*,c^*)\in Z^\circ$ in an affine chart of $\PP^n$ and a local
coordinate $t$ at $c^*$.  There $Z^\circ$ is the zero locus of the
single regular function $\widehat G(x,t)$, whose differential at
$(x^*,c^*)$ is nonzero --- already its $x$-part
$d_x\widehat G_{c^*}(x^*)$ is --- so $Z^\circ$ is smooth of dimension
$n$ at the point, with tangent space $\ker d\widehat G$.  The kernel
of the differential of the projection $\pi:Z^\circ\to C^\circ$ at the
point is
\[
  T_{(x^*,c^*)}Z^\circ\cap\{dt=0\}
  =\big\{\xi\ :\ d_x\widehat G_{c^*}(x^*)\cdot\xi=0\big\},
\]
of dimension $n-1$; hence $d\pi$ has rank one and is surjective onto
$T_{c^*}C^\circ$.  A morphism of smooth complex varieties whose
differential is surjective at every closed point is a smooth morphism,
here of relative dimension $n-1$.
\end{proof}

\begin{lemma}[Construction, domination, flatness]\label{lem:family}
$\Gfam^\circ$ has finitely many irreducible components
$E_1,\dots,E_r$.  After enlarging $\Sigma$ by finitely many points,
the reduced closed set
\[
  W\ :=\ \bigcup_{i:\ \pi(E_i)\ \mathrm{dense\ in}\ C}
  \overline{E_i}
  \ \subseteq\ \PP^n\times(\PP^n)^\vee\times C
\]
(closure in the ambient, $\pi$ the projection to $C$) satisfies:
every irreducible component of $W$ dominates $C$; $W$ is pure of
dimension $n$; and $\pi:W\to C$ is flat.
\end{lemma}

\begin{roadmapbox}{why domination gives flatness over a curve}
After vertical components are removed, every irreducible component of $W$
maps densely to $C$.  At a point $c\in C$, let $t_c$ be a local parameter.
Dominance says that $t_c$ does not vanish identically on any component.  Since
$W$ is reduced, this means $t_c$ is a non-zero-divisor.  Thus each local ring
of $W$ is torsion-free over the discrete valuation ring $\OO_{C,c}$, and a
torsion-free module over a DVR is flat.  The full proof also checks purity
of dimension.
\end{roadmapbox}

\begin{proof}
Noetherianity gives finitely many $E_i$.  Each image
$\pi(E_i)\subseteq C^\circ$ is constructible (Chevalley) in a curve,
hence finite or cofinite; put the points of the finite images into
$\Sigma$, so that over $U$ the fibres of $\Gfam^\circ$ and of
$\bigcup_{\mathrm{dom}}E_i$ agree.  (Note that the closure of a
non-dominating component lies over the same finite image, hence over
$\Sigma$.)  Every component of $W$ is some $\overline{E_i}$ with dense
image, and closure preserves domination.  Each dominating $E_i$ has
fibres of pure dimension $n-1$ over a dense subset of $C^\circ$
(Lemma~\ref{lem:smoothgraph}), so $\dim\overline{E_i}=n$ and $W$ is
pure of dimension $n$.  Flatness: $W$ is reduced, so the zero-divisors
of $\OO_W$ are the functions vanishing on a component; for a closed
point $c\in C$ with uniformizer $t_c\in\OO_{C,c}$ ($C$ smooth, so this
local ring is a DVR), $\pi^\#t_c$ vanishes on no component (each
dominates), hence is a non-zero-divisor in $\OO_{W,w}$ for every $w$
over $c$; thus $\OO_{W,w}$ is a torsion-free, hence flat,
$\OO_{C,c}$-module, and flatness is local
\cite[III.9.7]{Hartshorne77}.
\end{proof}

\begin{lemma}[Generic fibres are the Gauss-graph cycles, with
coefficient one]\label{lem:genericfibers}
After a further finite enlargement of $\Sigma$: for every $c\in U$ the
scheme-theoretic fibre $W_c$ has support $\supp\Gamma_c$, and
$[W_c]=\Gamma_c$ as cycles, with every coefficient equal to one.
\end{lemma}

\begin{roadmapbox}{why the generic fibre has exactly coefficient one}
Two facts must be checked separately.  First, no entire fibre component is
created by the boundary of a moving component; deleting finitely many
parameters ensures that the open graph is dense in every generic fibre
component.  Second, at the generic point of such a component the map to $C$
is smooth, so the fibre is reduced there.  The cycle coefficient is the length
of that generic local ring and is therefore one.
\end{roadmapbox}

\begin{proof}
\emph{Supports.}  Fix a dominating $E_i$.  By flatness over the smooth
curve, every component of every fibre $(\overline{E_i})_c$ has
dimension exactly $n-1$ \cite[III.9.5--9.6]{Hartshorne77}.  The
boundary $\overline{E_i}\smallsetminus E_i$ is closed of dimension
$\le n-1$, so over a cofinite set of $c$ its fibres have dimension
$\le n-2$; enlarge $\Sigma$ accordingly.  Then no fibre component of
$(\overline{E_i})_c$ lies in the boundary, i.e.\ $(E_i)_c$ is dense in
$(\overline{E_i})_c$ and $(\overline{E_i})_c=\overline{(E_i)_c}$.
Since the non-dominating components of $\Gfam^\circ$ and their
closures lie over $\Sigma$, while
$\bigcup_i(E_i)_c=\Gfam^\circ_c$ for $c\in U$, summing over the
dominating $i$ gives
\[
  \supp W_c\ =\ \bigcup_i\overline{(E_i)_c}
  \ =\ \overline{\Gfam^\circ_c}\ =\ \supp\Gamma_c .
\]

\emph{Coefficients.}  Enlarge $\Sigma$ further so that for $c\in U$
the fibres of the pairwise intersections
$\overline{E_i}\cap\overline{E_j}$ ($i\ne j$, both dominating), each
of dimension $\le n-1$, have dimension $\le n-2$.  Let $T$ be an
irreducible component of $W_c$; by the support computation, $T$ is a
component of $\overline{(E_i)_c}$ for some dominating $i$, and by the
dimension counts above $T$ contains a point $p$ lying in $(E_i)_c$,
outside the boundary $\overline{E_i}\smallsetminus E_i$, and outside
every $\overline{E_j}$ with $j\ne i$.  Near $p$ the scheme $W$
coincides with $\overline{E_i}$ and contains the subset $E_i$, which
is open in $\overline{E_i}$ near $p$ (as $\Gfam^\circ$ is open in its
closure and the other components of $\Gfam^\circ$ have been avoided);
on this open set the projection to $C$ is smooth
(Lemma~\ref{lem:smoothgraph}).  Hence $W\to C$ is smooth at $p$, so
the scheme-theoretic fibre $W_c$ is smooth --- in particular reduced
--- at $p$.  Reducedness passes to generizations:
$\OO_{W_c,\eta_T}$ is a localization of $\OO_{W_c,p}$, and
localizations of reduced rings are reduced.  So $W_c$ is reduced at
the generic point of $T$, and the coefficient of $T$ in the fibre
cycle $[W_c]$ --- the length of that local ring --- equals one.  Since
the supports agree and $\Gamma_c$ carries every component with
coefficient one by definition, $[W_c]=\Gamma_c$.
\end{proof}

\begin{lemma}[Conservation of multidegree]\label{lem:conservation}
$[W_0]$ is an effective cycle of pure dimension $n-1$, and for every
$k$ and every $c\in U$,
\[
  \pdeg_k\big([W_0]\big)=\pdeg_k\big(\Gamma_c\big).
\]
\end{lemma}

\begin{roadmapbox}{from flatness to conservation of every multidegree}
The special fibre has dimension $n-1$ and positive generic-length
multiplicities.  Proper flatness keeps the Euler characteristic of every
bidegree twist $\OO_{W_c}(a,b)$ constant.  The highest-degree coefficients
of this two-variable numerical polynomial are precisely the multidegrees
of the fibre cycle, with its scheme-theoretic multiplicities.
\end{roadmapbox}

\begin{proof}
Put $d=n-1$.  Each irreducible component of $W$ has dimension $n$ and
dominates $C$.  A fibre is cut out locally by a nonzero parameter on the
smooth curve, so every irreducible component of its support has dimension
$d$.  Its coefficient in the fundamental fibre cycle is the positive
finite length of the fibre's local ring at the generic point.  This proves
purity and effectivity, without requiring the special fibre to be reduced.

For conservation, let $\mathcal L_1$ and $\mathcal L_2$ be the restrictions
to $W$ of $\OO(1,0)$ and $\OO(0,1)$ from the fixed biprojective ambient.
The morphism $\pi:W\to C$ is proper and of finite presentation, and is flat
by Lemma~\ref{lem:family}.  Every invertible sheaf
$\mathcal L_1^{\otimes a}\otimes\mathcal L_2^{\otimes b}$ is therefore flat
over $C$.  Constancy of Euler characteristic in a proper flat family
\cite[Tag 0B9T]{StacksProject} implies that
\[
  \chi\!\left(W_c,\OO_{W_c}(a,b)\right)
\]
is independent of $c$, for every pair of integers $a,b$: the curve $C$ is
irreducible, hence connected.

For each $c$, these values are represented by a numerical polynomial
$P_c(a,b)\in\mathbb Q[a,b]$ of total degree at most $d$
\cite[Tag 0BEM]{StacksProject}.  Equality on all nonnegative integer pairs
implies $P_c=P_0$.  The homogeneous part of total degree $d$ is
\[
  (P_c)_d(a,b)=
  \sum_{k=0}^{d}
  \frac{\pdeg_k([W_c])}{(d-k)!\,k!}\,a^{d-k}b^k.
\]
This coefficient formula uses the fundamental cycle of the actual
scheme-theoretic fibre: its generic local lengths occur as multiplicities
\cite[Tags 0BEN, 0BEP, 0BES]{StacksProject}.  Comparing the displayed
coefficients gives $\pdeg_k([W_c])=\pdeg_k([W_0])$ for every $k$.
Lemma~\ref{lem:genericfibers} gives $[W_c]=\Gamma_c$ on $U$, as required.
Both twisting variables are needed; a single Hilbert polynomial for one
fixed Segre embedding would not by itself determine every multidegree.
\end{proof}

\begin{checklistbox}{the specialization package}
At this point the proof has a proper flat family $W\to C$, the generic fibre
cycle is exactly $\Gamma_c$ with coefficient one, and the special fibre is an
effective cycle with the same multidegrees.  What has \emph{not} yet been
proved is that $\Con(X)$ occurs inside $[W_0]$; that geometric containment is
the sole task of the next section.
\end{checklistbox}

\section{Why the Fermat conormal appears in the limit}
\label{sec:capture}

\begin{readerbox}{The only delicate point}
It is not enough that tangent hyperplanes have a formal limit: the approximating points must lie on multiplicity-one components, because $\Gone$ ignores repeated components.  At every smooth affine point of the Fermat cone, a transverse one-variable slice and Hurwitz's theorem produce nearby zeros with nonzero differential.  Those zeros are automatically on multiplicity-one branches, and their Gauss points converge to the desired conormal point.
\end{readerbox}

\paragraph{Proof idea.}
There are three logically distinct tasks.
\begin{enumerate}[leftmargin=2.2em,itemsep=2pt]
\item \textbf{Move the zero.}  The fixed point $x^*$ need not remain on the
nearby hypersurface, so find nearby points $x_c$ with $G_c(x_c)=0$.
\item \textbf{Keep the differential nonzero.}  Show
$d\widehat G_c(1,x_c)\ne0$ and identify its projective limit.  This places the
approximating points in the open Gauss graph, rather than only in its closure.
\item \textbf{Land in the moving family.}  Because $c\in U$, the component
containing each approximating point cannot be vertical.  Its closure is one of
the components used to define $W$, so the limit point lies in $W_0$.
\end{enumerate}

Fix a smooth affine zero $x^*$ of $F_n$.  Restrict $G_c$ to a line through $x^*$ on which $F_n$ has a simple zero.  That zero persists for small $c$, gradients converge, and Euler's identity controls the homogenizing coordinate of the covector.  Closing the dominating Gauss-graph components then places a dense open subset of $\Con(X)$ inside the special fibre $W_0$.

\begin{proposition}[Containment]\label{prop:containment}
$\Con(X)\subseteq\supp[W_0]$, and the coefficient of $\Con(X)$ in the
effective cycle $[W_0]$ is at least one.
\end{proposition}

The proof must produce points of the \emph{multiplicity-one} Gauss
graphs $\Gone(\widehat G_c)$ converging into the Fermat conormal; this
is the reason for using $\Gone$, since the limit divisor
$V(\widehat G_0)=V(x_0^{\,m-n}\widetilde F_n)$ is badly non-reduced at
infinity and nothing guarantees the nearby divisors
$V(\widehat G_c)$ are reduced either.

\begin{lemma}[The limit covector at a smooth affine cone point]
\label{lem:covector}
Let $x^*\in V(F_n)\subset\CC^n$, $x^*\ne0$, so that
$\nabla F_n(x^*)=n\,(x_i^{*\,n-1})_i\ne0$.  Then
\[
  d\widehat G_0(1,x^*)\ =\ \big(0,\ \nabla F_n(x^*)\big)\ \ne\ 0,
\]
whose projective class $[0:\nabla F_n(x^*)]$ is the conormal covector
of the cone $X$ at the smooth point $(1,x^*)$.
\end{lemma}

\begin{roadmapbox}{compute the special tangent covector directly}
Differentiate the product $x_0^{m-n}\widetilde F_n$.  At a point of the
Fermat component in the chart $x_0=1$, every derivative of the padding factor
is multiplied by $\widetilde F_n$ and therefore vanishes.  The remaining
entries are exactly the affine gradient of $F_n$, with zero in the $x_0$
coordinate.
\end{roadmapbox}

\begin{proof}
With $\widehat G_0=x_0^{\,m-n}\widetilde F_n$
(Lemma~\ref{lem:specialfiber}):
\[
  \partial_0\widehat G_0=(m-n)\,x_0^{\,m-n-1}\widetilde F_n
  +x_0^{\,m-n}\,\partial_0\widetilde F_n,
  \qquad
  \partial_i\widehat G_0=x_0^{\,m-n}\,\partial_i\widetilde F_n\ \
  (i\ge1).
\]
At $(1,x^*)$: the first term of $\partial_0$ vanishes because
$\widetilde F_n(1,x^*)=F_n(x^*)=0$, the second because
$\widetilde F_n$ does not involve $x_0$, and
$\partial_i\widehat G_0(1,x^*)=\partial_iF_n(x^*)$.  The hyperplane
$\{\sum_{i\ge1}\partial_iF_n(x^*)\,y_i=0\}$ is the projective tangent
hyperplane of $X$ at $(1,x^*)$ and contains the vertex $q_0$ --- the
$\xi_0=0$ signature of a cone.  Note the mechanism: the non-reduced
factor $x_0^{\,m-n}$ contributes to the differential only through the
term carrying $\widetilde F_n$ itself, which dies \emph{because the
differential is evaluated on the zero locus of $\widetilde F_n$}.
Padding is differentially invisible along the Fermat component.
\end{proof}

\begin{lemma}[Hurwitz persistence onto multiplicity-one branches]
\label{lem:hurwitz}
For every $x^*$ as in Lemma~\ref{lem:covector} there are parameters
$c_k\in U$ with $c_k\to0$ analytically and points
$x_k\to(1,x^*)$ in the chart $x_0=1$ such that
\[
  \widehat G_{c_k}(x_k)=0,\qquad d\widehat G_{c_k}(x_k)\ne0,\qquad
  \big[d\widehat G_{c_k}(x_k)\big]\ \longrightarrow\
  \big[0:\nabla F_n(x^*)\big].
\]
Moreover each $x_k$ lies on a multiplicity-one component of
$V(\widehat G_{c_k})$, so
$\big(x_k,[d\widehat G_{c_k}(x_k)]\big)$ lies in the \emph{open graph}
$\Gfam^\circ_{c_k}$ defining $\Gamma_{c_k}$ --- not merely in its
closure.
\end{lemma}

\begin{roadmapbox}{a one-variable implicit-function argument}
Choose a direction $w$ on which the nonzero differential
$dF_n(x^*)$ evaluates to one.  The restricted function
$\varphi_0(s)=F_n(x^*+sw)$ then has a simple zero at $s=0$.  Uniform
coefficient convergence gives uniform convergence of $\varphi_c$ and its
derivative on a small disk.  Hurwitz's theorem preserves the simple zero,
producing $s_c\to0$.  Gradient convergence keeps the full differential
nonzero, and Euler's identity supplies the missing homogeneous coordinate.
\end{roadmapbox}

\begin{proof}
Work in the chart $x_0=1$, where $\widehat G_c=G_c$.  Choose
$w\in\CC^n$ with $\langle dF_n(x^*),w\rangle=1$ --- possible since the
functional is nonzero; one may \emph{not} simply take
$w=\nabla F_n(x^*)$, which can be isotropic over $\CC$.  Define
$\varphi_c(s):=G_c(x^*+sw)$, holomorphic in $s$ on a disk
$\overline{D_\rho}$, with coefficients regular (hence analytically
continuous) in $c$.  Then $\varphi_0(s)=F_n(x^*+sw)$ has a simple zero
at $s=0$: $\varphi_0(0)=0$, $\varphi_0'(0)=1$.  Shrink $\rho$ so that
$s=0$ is the only zero of $\varphi_0$ in $\overline{D_\rho}$.  Since
$G_c\to F_n$ coefficientwise as $c\to0$ (regular functions on $C$ are
analytically continuous and take the value at $0$ prescribed by
Lemma~\ref{lem:normalform}(a)), $\varphi_c\to\varphi_0$ uniformly on
$\overline{D_\rho}$; by Hurwitz's theorem --- if holomorphic
$f_k\to f\not\equiv0$ uniformly on $\overline{D_\rho}$, $f$ zero-free
on $\partial D_\rho$ with exactly one zero counted with multiplicity
inside, then so is $f_k$ for large $k$, with its zero converging to
that of $f$ (argument principle) --- there is, for each $c$ small in
$U$, a unique $s_c\in D_\rho$ with $G_c(x^*+s_cw)=0$ and $s_c\to0$.
Set $x_c:=x^*+s_cw$.

Gradients: $\nabla G_c\to\nabla F_n$ uniformly on compacts (again
coefficientwise convergence of polynomials of bounded degree), so
$\nabla G_c(x_c)\to\nabla F_n(x^*)\ne0$, and for small $c$ the affine
entries $\partial_i\widehat G_c(1,x_c)=\partial_iG_c(x_c)$, $i\ge1$,
are not all zero, whence $d\widehat G_c(1,x_c)\ne0$.  The
$0$-th entry is determined by Euler's relation for the degree-$m$ form
$\widehat G_c$:
\[
  \partial_0\widehat G_c(1,x_c)
  =m\,G_c(x_c)-\sum_{i\ge1}x_{c,i}\,\partial_iG_c(x_c)
  \ \longrightarrow\
  0-\langle x^*,\nabla F_n(x^*)\rangle=-n\,F_n(x^*)=0 ,
\]
so $[d\widehat G_c(1,x_c)]\to[0:\nabla F_n(x^*)]$.

Multiplicity-one membership --- the point of the construction: $x_c$
is a point of $V(\widehat G_c)$ at which the differential does not
vanish.  More explicitly, write a local factorization
$\widehat G_c=u\prod_j g_j^{a_j}$ with $u$ a unit.  If the point lies on a
factor with $a_j\ge2$, every term in the derivative contains $g_j$ and
vanishes.  If it lies on two distinct components, each product-rule term
contains the equation of the other component and again vanishes.  Therefore a
nonzero differential forces the point to lie on exactly one local component,
and that component has multiplicity one.  A point lying on a
multiplicity-$\ge2$ component has
$d\widehat G_c=0$ there (Lemma~\ref{lem:gonestructure}), so $x_c$ lies
only on multiplicity-one components and, having nonvanishing
differential, lies in the open graph.  No squarefreeness of
$\widehat G_c$ is assumed; whatever multiple components exist live
elsewhere.  Finally, restrict the parameters to a sequence
$c_k\in U$, $c_k\to0$, which exists since $\Sigma$ is finite.
\end{proof}

\begin{roadmapbox}{close a dense set of conormal points}
For each smooth affine point of the Fermat cone, Lemma~\ref{lem:hurwitz}
produces a convergent sequence in the moving open graph.  Closedness of $W$
puts the limit in $W_0$.  These limits cover the conormal graph over a dense
open subset of $X$, so their closure is all of $\Con(X)$.  Finally, an
irreducible subvariety of the same dimension as the effective fibre cycle is
an actual fibre component and therefore has a positive integer coefficient.
\end{roadmapbox}

\begin{proof}[of Proposition~\ref{prop:containment}]
The points constructed in Lemma~\ref{lem:hurwitz} lie in
$\Gfam^\circ_{c_k}$ with $c_k\in U$; the component of $\Gfam^\circ$
carrying such a point dominates $C$ (a non-dominating component has
image inside $\Sigma$, which $U$ avoids), so the point lies in $W$.
Since $W$ is Zariski-closed, it is closed in the analytic topology,
and the analytic limit
$\big((1,x^*),[0:\nabla F_n(x^*)]\big)$ lies in
$W\cap\pi^{-1}(0)=W_0$.  As $x^*$ ranges over
$V(F_n)\smallsetminus\{0\}$, these limit points range over the
conormal graph of $X$ over the smooth affine cone points, which is
dense in $\Con(X)$: the smooth locus of $X$ is
$X\smallsetminus\{q_0\}$ (Proposition~\ref{prop:coneshift}(a)), its
intersection with the chart $x_0\ne0$ is dense in $X$ ($X$ is
irreducible and $X\not\subseteq\{x_0=0\}$), and the conormal covector
over a smooth point is unique.  Hence
$\Con(X)\subseteq\supp[W_0]$.

Coefficient: $[W_0]$ is effective of pure dimension $n-1$
(Lemma~\ref{lem:conservation}); $\Con(X)$ is irreducible of dimension
$n-1$ and is contained in the support, hence is an irreducible
component of $\supp[W_0]$, and its coefficient is a positive integer.
\end{proof}

\begin{checklistbox}{why this proves component containment, not merely pointwise convergence}
The approximating points lie in components that dominate $C$; hence they are
contained in $W$, not in a discarded vertical component.  Their limits range
over a dense open part of the Fermat conormal.  Since $W_0$ is closed, it
contains the closure of that dense set.  Because both $\Con(X)$ and the fibre
components have dimension $n-1$, containment in the support makes
$\Con(X)$ itself an irreducible component with positive coefficient.
\end{checklistbox}

\begin{remark}[Additional components in the limit]\label{rem:extra}
Besides $\Con(X)$, the cycle $[W_0]$ may contain: conormal-type cycles
over the hyperplane at infinity (limits of graph points escaping to
$x_0=0$; for $c\in U$ the divisor $V(\widehat G_c)$ may itself contain
$V(x_0)$ as a component of multiplicity $m-\deg_xG_c$, which when that
multiplicity is one rides through the family with
$\pdeg_{n-2}=0$, cf.\ Remark~\ref{rem:degreem}, and when $\ge2$ is
invisible to $\Gone$); cycles of the form
$S'\times\{\text{linear space}\}$ over singular strata; and
Lagrangian-type components over the singular locus of the limit
divisor.  Effectivity is the key point: extra components enter
$[W_0]$ with nonnegative coefficients and nonnegative multidegrees
(Lemma~\ref{lem:kleiman}) and can only \emph{add}.  Nothing subtracts
the $\Con(X)$ summand, because cycle specialization has no negative
coefficients and Proposition~\ref{prop:containment} pins $\Con(X)$
inside the support with coefficient $\ge1$.  Contrast the two limit
objects: the \emph{divisor} $V(\widehat G_0)$ is badly non-reduced at
infinity; the \emph{cycle} $[W_0]$ is insensitive to that, precisely
because $\Gone$ is built from the differential, which
Lemma~\ref{lem:covector} shows survives along the Fermat component.
\end{remark}

\begin{corollary}[Special $\le$ generic]\label{cor:specialgeneric}
For every $c\in U$:
\[
  \pdeg_{n-2}\big(\Con(X)\big)\ \le\ \pdeg_{n-2}\big([W_0]\big)
  \ =\ \pdeg_{n-2}\big(\Gamma_c\big).
\]
\end{corollary}

\begin{proof}
Write $[W_0]=a\cdot\Con(X)+\sum_ja_jV_j$ with $a\ge1$, $a_j\ge0$
integers, and $V_j$ irreducible of dimension $n-1$
(Proposition~\ref{prop:containment},
Lemma~\ref{lem:conservation}).  Each $\pdeg_{n-2}(V_j)\ge0$
(Lemma~\ref{lem:kleiman}), so
$\pdeg_{n-2}([W_0])\ge a\,\pdeg_{n-2}(\Con(X))
\ge\pdeg_{n-2}(\Con(X))$.  The equality is
Lemma~\ref{lem:conservation}.
\end{proof}

\section{Cones and the next-to-top multidegree}
\label{sec:coneshift}

\begin{readerbox}{Why the index is $n-2$}
Every tangent hyperplane to a cone contains the vertex, so the top dual slot $\pdeg_{n-1}$ vanishes.  The next slot adds one generic condition on the point of the cone; this cuts each ruling line once and leaves exactly the dual-degree count of the base hypersurface.
\end{readerbox}

To read $\pdeg_{n-2}(\Con(X_Z))$ enumeratively, choose
\begin{itemize}[leftmargin=2em,itemsep=1pt]
\item a generic hyperplane $M\subset\PP^n$ in the primal factor; and
\item a generic codimension-$(n-2)$ linear space
$L\subset(\PP^n)^\vee$ in the dual factor.
\end{itemize}
First impose the dual condition.  It selects finitely many tangent
hyperplanes to the base $Z$, in number $\deg Z^\vee$.  Over each selected
base tangency, the cone still has one ruling line of possible primal points.
The hyperplane $M$ intersects that line once.  Thus the count is unchanged
when passing from the base conormal to the next-to-top conormal multidegree of
the cone.

\begin{proposition}[Cone shift]\label{prop:coneshift}
Let $Z\subset\PP^{n-1}=\{x_0=0\}$ be a smooth irreducible hypersurface
of degree $e\ge2$, $n\ge3$, and let $X_Z\subset\PP^n$ be the cone over
$Z$ with vertex $q_0=[1:0:\cdots:0]$.  Then:
\begin{enumerate}[label=(\alph*),leftmargin=2.2em]
\item $\Sing X_Z=\{q_0\}$, and
$\Con(X_Z)\subset\PP^n\times q_0^{\perp}$, where
$q_0^\perp=\{\xi_0=0\}\cong(\PP^{n-1})^\vee$;
\item $\pdeg_{n-1}(\Con(X_Z))=0$\ \ (the cone is dual-deficient);
\item $\pdeg_{n-2}(\Con(X_Z))=\deg Z^\vee$.
\end{enumerate}
In particular, for the Fermat cone $X$ ($Z=V(F_n)$ smooth since
$\nabla F_n$ vanishes only at the origin, $e=n$),
Lemma~\ref{lem:dualdeg} with $N=n-1$ gives
\[
  \pdeg_{n-2}\big(\Con(X)\big)\ =\ n(n-1)^{n-2}.
\]
\end{proposition}

\begin{roadmapbox}{prove the cone shift in three geometric moves}
Part~(a) identifies the smooth locus and shows that every tangent hyperplane
contains the vertex.  Part~(b) uses that restriction to explain why the top
dual multidegree vanishes.  Part~(c) imposes the next-to-top flag: the dual
conditions count tangent hyperplanes to the base, and the remaining primal
condition cuts each ruling line once.
\end{roadmapbox}

\begin{proof}
Write the defining form of $X_Z$ as $q(x_1,\dots,x_n)$, independent of
$x_0$.  (a)~Its differential is $(0,\nabla q)$.  Smoothness of $Z$ shows
that this differential is nonzero at every point of $X_Z$ except $q_0$;
at $q_0$ it vanishes because $e\ge2$.  The conormal covector at
$[x_0:z]$ with $z\ne0$ is $[0:\nabla q(z)]$.  Thus every tangent
hyperplane contains $q_0$, and taking closures gives
$\Con(X_Z)\subset\PP^n\times q_0^\perp$.

Radial projection
\[
  \rho:X_Z\smallsetminus\{q_0\}\longrightarrow Z,
  \qquad [x_0:z]\longmapsto[z],
\]
is a morphism.  On this open set the conormal is the graph of the displayed
covector, and $\rho$ identifies its dual coordinate with the Gauss covector
of $Z$.  This description is used only away from the vertex; no global
bundle map from $\Con(X_Z)$ to $\Con(Z)$ is required.  The open graph has
dimension $(n-2)+1=n-1$, and so does its closure.

(b)~The dual image is $Z^\vee\subset q_0^\perp$, which has dimension
$n-2$ by Lemma~\ref{lem:dualdeg}.  A generic collection of $n-1$ dual
hyperplanes misses this image.  Hence $\pdeg_{n-1}(\Con(X_Z))=0$.

(c)~Choose a generic primal hyperplane $M$ avoiding $q_0$.  Write its
equation as $\alpha x_0+\sum_i b_i x_i=0$, with $\alpha\ne0$.
Radial projection restricts to an isomorphism
\[
  \rho:X_Z\cap M\ \xrightarrow{\ \sim\ }\ Z,
  \qquad
  \rho^{-1}([z])=
  \left[-\alpha^{-1}\sum_i b_i z_i:z\right].
\]  All of this section
lies in the smooth locus of the cone.  On the conormal graph, the equation
of $M$ cuts the ruling direction with nonzero coefficient $\alpha$;
therefore it cuts with multiplicity one.  Consequently the scheme-theoretic
section
\[
  \Con(X_Z)\cap\bigl(M\times(\PP^n)^\vee\bigr)
\]
is isomorphic to $\Con(Z)$, preserving the dual coordinate under
$q_0^\perp\cong(\PP^{n-1})^\vee$.

Now impose $n-2$ generic dual hyperplanes.  Their restrictions to
$q_0^\perp$ are generic hyperplanes there.  By the preceding isomorphism,
the resulting intersection number is
\[
  \pdeg_{n-2}(\Con(X_Z))
  =\deg\!\left(\Con(Z)\cdot h_2^{n-2}\right)
  =\deg Z^\vee.
\]
For the last equality, the conormal projection to $Z^\vee$ is birational
in characteristic zero by Lemma~\ref{lem:dualdeg}; a generic line in
$(\PP^{n-1})^\vee$ meets its hypersurface image in $\deg Z^\vee$ points,
each with one conormal lift.  Generic transversality as in
Lemma~\ref{lem:kleiman} also realizes these intersection numbers as reduced
point counts.  This proves (c).
\end{proof}

\begin{remark}
Slot $n-1$ is blind (part (b)); slot $n-2$ is the first informative
multidegree of a cone and records exactly the dual degree of the base.
This is the global avatar of the companions' local invariant --- the
top polar degree of the projectivized tangent cone at the vertex
\emph{is} $\deg Z^\vee$ --- and it is the conversion of that local
invariant into a global multidegree, available only because $F_n$ is
homogeneous, that lets conservation
(Lemma~\ref{lem:conservation}) act on it.  For $n=3$, the identity gives $3\cdot2=6$, the classical dual degree of a smooth plane cubic.
\end{remark}

\section{Combining the lower and upper counts}\label{sec:proofs}

\begin{readerbox}{The entire numerical comparison}
For a generic determinantal fibre $\Gamma_c$,
\[
 n(n-1)^{n-2}
 =\pdeg_{n-2}(\Con(X))
 \le \pdeg_{n-2}(\Gamma_c)
 \le B(m,n)
 \le \left(\frac{4em}{n-1}\right)^{n-1}.
\]
The symmetric proof replaces the last term by $\left(\frac{2em}{n-1}\right)^{n-1}$.  Taking $(n-1)$-st roots gives the stated quadratic lower bounds.
\end{readerbox}

\begin{lemma}[Estimates]\label{lem:estimates}
For $n\ge3$ and $m\ge n$:
\[
  B(m,n)\ \le\ 2^{n-2}\binom{2m-1}{n-1}\ \le\
  \Big(\frac{4em}{n-1}\Big)^{n-1},
  \qquad
  2^{n-2}\binom{m}{n-1}\ \le\ \Big(\frac{2em}{n-1}\Big)^{n-1}.
\]
\end{lemma}

\begin{roadmapbox}{reduce the exact coefficient to a standard binomial estimate}
The exact sum $B(m,n)$ contains a factor $\binom{n-2}{i-1}$.  Bound this
factor by the sum of all binomial coefficients, $2^{n-2}$.  The remaining
convolution is a Vandermonde coefficient
$\binom{2m-1}{n-1}$.  Finally apply
$\binom Nk\le(eN/k)^k$ and make only the two harmless relaxations
$2m-1\le2m$ and $2^{n-2}\le2^{n-1}$.
\end{roadmapbox}

\begin{proof}
$\binom{n-2}{i-1}\le2^{n-2}$, and by Vandermonde
$\sum_{i}\binom{m}{i}\binom{m-1}{n-1-i}=\binom{2m-1}{n-1}$, of which
$B(m,n)/2^{n-2}$ retains a sub-sum of nonnegative terms.  Then
$\binom{N}{k}\le(eN/k)^k$ gives
$2^{n-2}\binom{2m-1}{n-1}
\le2^{n-2}\big(\tfrac{e(2m-1)}{n-1}\big)^{n-1}
\le2^{n-1}\big(\tfrac{2em}{n-1}\big)^{n-1}
=\big(\tfrac{4em}{n-1}\big)^{n-1}$, and similarly
$2^{n-2}\binom{m}{n-1}\le2^{n-1}\big(\tfrac{em}{n-1}\big)^{n-1}
=\big(\tfrac{2em}{n-1}\big)^{n-1}$.
\end{proof}

\begin{detailbox}{Trace of every comparison in the main chain}
\begin{tabularx}{\textwidth}{@{}p{0.31\textwidth}X@{}}
$n(n-1)^{n-2}=\pdeg_{n-2}(\Con(X))$ & Cone-shift identity and Fermat dual degree. \\
$\pdeg_{n-2}(\Con(X))\le\pdeg_{n-2}(\Gamma_c)$ & $\Con(X)$ is an effective component of $[W_0]$, and multidegrees are conserved. \\
$\pdeg_{n-2}(\Gamma_c)\le B(m,n)$ & Determinantal conormal bound applied to the exact generic fibre. \\
$B(m,n)\le(4em/(n-1))^{n-1}$ & Lemma~\ref{lem:estimates}. \\
$(n-1)^{n-1}\le n(n-1)^{n-2}$ & The elementary inequality $n\ge n-1$. \\
\end{tabularx}
\end{detailbox}

\begin{proof}[of Theorem~\ref{thm:border}]
Suppose $\dcb(F_n)\le m$.  Run
Lemmas~\ref{lem:normalform}--\ref{lem:specialfiber}: a smooth curve
$C$, the family $\widehat G$ homogenized at degree $m$
(so $m\ge n$ by Lemma~\ref{lem:mgen}), fiberwise identities
$\widehat G_c=\det\widehat A_c$ on $C^\circ$, special fibre
$x_0^{\,m-n}\widetilde F_n$ with multiplicity-one Fermat component.
Build the family $W$ and fix any $c\in U$ ($U$ is infinite, so such
$c$ exists).  Chain the results:
\[
  n(n-1)^{n-2}
  \ \overset{\ref{prop:coneshift}}{=}\
  \pdeg_{n-2}\big(\Con(X)\big)
  \ \overset{\ref{cor:specialgeneric}}{\le}\
  \pdeg_{n-2}\big(\Gamma_c\big)
  \ \overset{\ref{thm:detlemma}\mathrm{(i)}}{\le}\
  B(m,n)
  \ \overset{\ref{lem:estimates}}{\le}\
  \Big(\frac{4em}{n-1}\Big)^{n-1},
\]
the middle application of Theorem~\ref{thm:detlemma}(i) being
directly: $\Gamma_c=\Gone(\det\widehat A_c)$ with
$\det\widehat A_c=\widehat G_c\ne0$ (Lemmas~\ref{lem:fiberdet}
and~\ref{lem:nonzero}), and the theorem's only hypothesis is exactly
that.  Now $n(n-1)^{n-2}\ge(n-1)\cdot(n-1)^{n-2}=(n-1)^{n-1}$, so
taking $(n-1)$-st roots of
$(n-1)^{n-1}\le\big(\tfrac{4em}{n-1}\big)^{n-1}$ gives
$n-1\le\tfrac{4em}{n-1}$, i.e.\ $m\ge\tfrac{(n-1)^2}{4e}$.
Every use of a generic choice occurred after the fixed fibre
$\widehat A_c$ was selected; no genericity of the determinantal
representation itself is required.
\end{proof}

\begin{proof}[of Theorem~\ref{thm:symborder}]
Identical, with $D_m^{\mathrm{sym}}$ in
Lemma~\ref{lem:normalform} (the image of the symmetric matrix space is
likewise irreducible constructible, and curve selection is
representation-agnostic), symmetric fiberwise representations $A_c$,
$\widehat A_c$ symmetric (homogenization preserves symmetry,
Lemma~\ref{lem:fiberdet}), Theorem~\ref{thm:detlemma}(ii) in the
chain, and the second estimate of Lemma~\ref{lem:estimates}:
$(n-1)^{n-1}\le\big(\tfrac{2em}{n-1}\big)^{n-1}$, so
$m\ge\tfrac{(n-1)^2}{2e}$.  Sections~\ref{sec:family}
--\ref{sec:coneshift} are representation-agnostic --- they see only
the polynomial family --- and need no change.  The model is non-vacuous:
$\sdcb(F_n)\le\sdc(F_n)\le2n^2+2n+1$
\cite[Prop.~5.1]{SheshadriSDC}.
\end{proof}

\begin{remark}[No hidden logarithmic loss; where the slack is]
\label{rem:nolog}
The bounds $(n-1)^2/(4e)$ and $(n-1)^2/(2e)$ hold as stated for every
$n\ge3$.  In the form $(\frac{1}{4e}-o(1))n^2$, the $o(1)$ is only the
identity $(n-1)^2/n^2=1-O(1/n)$.  The bounds are not the exact
thresholds.  Let $m^*(n)$ be the least $m$ with
$n(n-1)^{n-2}\le B(m,n)$.  Then
$m^*(n)=\bigl(1+O(\frac{\log n}{n})\bigr)\,(n-1)^2/(4e)$; at $n=320$, for
example, $m^*(n)=9676$ while $(n-1)^2/(4e)\approx9359$, a ratio of
$1.034$.  The symmetric threshold behaves the same way (ratio $1.023$ at
$n=320$).  Nearly all of this slack comes from
Lemma~\ref{lem:estimates}: the estimate $\binom Nk\le(eN/k)^k$ drops a
Stirling factor of order $\sqrt k$, and the relaxations of $B(m,n)$ to
$2^{n-2}\binom{2m-1}{n-1}$ and of $2^{n-2}$ to $2^{n-1}$ drop further
polynomial factors.  After the $(n-1)$st root, each of these costs a
factor $1+O(\frac{\log n}{n})$.  By contrast, the absorptions
$n(n-1)^{n-2}\ge(n-1)^{n-1}$ and $2m-1\le2m$ cost only a factor
$(n/(n-1))^{1/(n-1)}=1+O(1/n^2)$ and an additive $\tfrac12$.  The exact
thresholds therefore approach $1/(4e)=0.09197\ldots$ and
$1/(2e)=0.18394\ldots$, in units of $(n-1)^2$, from above.
\end{remark}

\begin{proof}[of Corollary~\ref{cor:exact}]
$\dc\ge\dcb$ and $\sdc\ge\sdcb$ by definition.  Alternatively and more
directly, Remark~\ref{rem:recover} derives the exact inequalities from
Theorem~\ref{thm:detlemma} and Proposition~\ref{prop:coneshift} alone,
with no specialization machinery.
\end{proof}

\section{What the argument uses, and what it does not}\label{sec:scope}

The proof is robust to singular and non-reduced determinantal fibres, but it is not a universal lower-bound machine.  The following list separates genuine hypotheses from technical conveniences.

\begin{enumerate}[leftmargin=2em]
\item \emph{$\Gone$ throughout.}  Every fiberwise object is the
multiplicity-one Gauss-graph cycle, never the full reduced conormal of
the fibre.  Theorem~\ref{thm:detlemma} bounds exactly $\Gone$; whether
the reduced conormal of a non-squarefree determinant obeys the same
bound is open (Remark~\ref{rem:gammared}) and is not needed.
\item \emph{Why the Fermat conormal survives.}  Smooth affine zeros of
$F_n$ persist as zeros of $G_c$ with nonvanishing differential
(Hurwitz plus uniform $C^1$ convergence), and a nonvanishing
differential forces membership in multiplicity-one components; so the
approximating points are honest open-graph points, and properness
pushes their limits into $W_0$.  The padding $x_0^{\,m-n}$ is
differentially dead along the Fermat component
(Lemma~\ref{lem:covector}) and contributes only additional effective components
(Remark~\ref{rem:extra}).
\item \emph{No genericity, and no compatibility, of the
representations.}  The $A_c$ are arbitrary, and no relation among them
across different $c$ is assumed; the only object with structure along
the parameter is the polynomial family $G$.  The only generic choices
in the paper are flags (Lemma~\ref{lem:kleiman}, used in
Theorem~\ref{thm:detlemma} and Proposition~\ref{prop:coneshift}) and
the constant matrix $\Lambda$ of Step~3, all chosen \emph{after} the
data.
\item \emph{Repeated factors are allowed everywhere.}
Multiplicity-$\ge2$ components of any fibre, including a hyperplane
at infinity of multiplicity $m-\deg_xG_c$, are invisible to $\Gone$
on both sides of the chain (Remarks~\ref{rem:excess},
\ref{rem:degreem}, \ref{rem:extra}); a multiplicity-one hyperplane at
infinity has $\pdeg_{n-2}=0$ and is harmless.
\item \emph{Transversality and flatness inventory.}  Kleiman
transversality: Sections~\ref{sec:detlemma}
and~\ref{sec:coneshift}.  Fulton positivity / localized Chern
classes: Lemma~\ref{lem:positivity}.  Curve selection:
Lemma~\ref{lem:normalform}.  Smoothness of the open graph over the
curve (Jacobian criterion) and reducedness of the generic fibre
cycles: Lemmas~\ref{lem:smoothgraph} and~\ref{lem:genericfibers}.
Torsion-freeness over a DVR and flat-fibre dimension theory:
Lemmas~\ref{lem:family} and~\ref{lem:genericfibers}.  Specialization
of cycle classes: Lemma~\ref{lem:conservation}.  Hurwitz:
Lemma~\ref{lem:hurwitz}.
\item \emph{Homogeneity is essential.}  The cone-shift identity
consumes homogeneity of the target: it converts the vertex-local
invariant into the global multidegree that specialization conserves.
For an inhomogeneous target no such bridge is provided, and the method
as written says nothing.
\item \emph{No permanent claim.}  For the permanent the special-side
input would be a lower bound on the relevant conormal multidegree of
the (singular) permanent hypersurface, which we do not know; the
transfer machinery of Sections~\ref{sec:normalform}
--\ref{sec:capture} would apply once such a bound exists, since
nothing there is specific to $F_n$ except homogeneity and
Lemma~\ref{lem:covector}'s use of the explicit gradient (which
generalizes to any reduced homogeneous target along its smooth affine
locus).
\item \emph{Base changes and shrinkings.}  One normalization
(Lemma~\ref{lem:normalform}) and finitely many enlargements of the
excluded set $\Sigma$ (Lemmas~\ref{lem:normalform},
\ref{lem:nonzero}, \ref{lem:family}, \ref{lem:genericfibers}).  None
alters the hypothesis $F_n\in\overline{D_m}$ or the parameter $m$; the
final chain needs a single $c\in U$, and $U$ is infinite.
\item \emph{Characteristic zero.}  Used in Kleiman transversality,
reflexivity (Lemma~\ref{lem:dualdeg}), the symmetric quadratic-form
argument in Step~4 of Section~\ref{subsec:symproof}, and the analytic
Hurwitz argument (replaceable by a Newton--Puiseux argument over
$\CC[[t]]$, which we do not pursue).  The statement itself fails as a
uniform theorem in positive characteristic: for
$\operatorname{char}=p>2$ and $n=p^r$ the Frobenius identity makes
$F_n$ a perfect power of a linear form, with $\sdc(F_n)\le n$
\cite[Rem.~6.1]{SheshadriSDC}.
\end{enumerate}

\paragraph{Open questions.}
(1)~Improve the constants toward the upper bounds; the B\'ezout count
is an overcount and the limit cycle $[W_0]$ carries identifiable additional components
(Remark~\ref{rem:extra}) whose multidegrees could in principle be
subtracted.  (2)~Bound the relevant conormal multidegree of the
permanent hypersurface from below, activating item~7.
(3)~Settle Remark~\ref{rem:gammared}.  (4)~Is $\dcb(F_n)<\dc(F_n)$
for some $n$?  Both are now pinned to $\Theta(n^2)$ with the same
lower-bound constant, and the question of a genuine border gap for
this family becomes meaningful.

\appendix

\section{Glossary for readers outside algebraic geometry}\label{app:glossary}

This appendix records the precise informal meaning of the terms that recur in
the proof.  The formal definitions and citations remain in the main text.

\begin{description}[style=nextline,leftmargin=3.0cm,labelwidth=2.7cm,itemsep=3pt]
\item[Affine chart] The subset $x_0\ne0$ of projective space, identified with
ordinary affine space by setting $x_0=1$.
\item[Conormal variety] The closure of pairs $(x,H)$ where $x$ is a smooth
point of a variety and $H$ is a hyperplane containing the tangent space at
$x$.  For a hypersurface, $H$ is represented by the projective gradient.
\item[Cycle] A formal integer combination of irreducible subvarieties of the
same dimension.  The coefficient records geometric multiplicity.
\item[Effective cycle] A cycle all of whose coefficients are nonnegative.
This property makes multidegree inequalities monotone componentwise.
\item[Dominating component] An irreducible component whose image is dense in
the parameter curve.  Such a component occurs for a moving set of fibres.
\item[Vertical component] A component whose image is not dense; over a curve,
its image is finite after constructibility is used.
\item[Flat family] A family with no algebraic torsion in the parameter.  In
this paper, reducedness plus domination over a smooth curve makes the local
rings torsion-free over a DVR and hence flat.  Flatness is used to specialize
fibre cycles with multiplicities preserved.
\item[Generic choice] A choice outside a proper Zariski-closed exceptional
set.  Only auxiliary flags and the map $\Lambda$ are chosen generically; the
input determinant is arbitrary.
\item[Gauss graph] The graph of the map sending a smooth hypersurface point to
its tangent hyperplane.
\item[Multiplicity-one component] An irreducible factor occurring to exponent
one in the defining equation.  At its generic smooth points the differential
is nonzero.
\item[Multidegree] An intersection number in a product of projective spaces.
Here $\pdeg_k$ counts $n-1-k$ generic primal conditions and $k$ generic dual
conditions.
\item[Properness] The algebro-geometric compactness property.  A closed
subvariety of a projective bundle over the parameter curve is proper over the
curve, so limits remain in the total space.
\item[Reduced isolated point] A zero-dimensional local component whose local
ring is $\CC$; its intersection multiplicity is one.
\item[Scheme-theoretic fibre] The fibre defined by the parameter equation in
the structure sheaf, including nilpotents and multiplicities rather than only
the underlying point set.
\item[Specialization of cycles] The passage from a generic fibre cycle to the
cycle of a special fibre, with its local-length multiplicities.  In our proper
flat family in a fixed product of projective spaces, every multidegree is
conserved; Lemma~\ref{lem:conservation} proves this using all bidegree twists.
\end{description}

\section{Standard geometric inputs and exactly where they are used}
\label{app:inputs}

\begin{center}
\small
\begin{tabularx}{\textwidth}{@{}p{0.25\textwidth}p{0.28\textwidth}X@{}}
\toprule
Input & Statement used & Location and role \\
\midrule
Curve selection & A closure point of a constructible set lies on an algebraic curve meeting the set densely & Lemma~\ref{lem:normalform}; replaces abstract border membership by a one-parameter polynomial family \\
Normalization of a curve & The normalized curve is smooth, with DVR local rings at closed points & Lemma~\ref{lem:normalform} and Lemma~\ref{lem:family}; supplies a regular parameter and the torsion-free criterion \\
Kleiman transversality & A generic translate meets prescribed varieties transversely in their chosen dense opens & Step~1 of Theorem~\ref{thm:detlemma} and Proposition~\ref{prop:coneshift}; converts multidegrees into reduced point counts \\
Localized top Chern class & Isolated zeros of a section consume positive local multiplicity bounded by the top Chern number & Lemma~\ref{lem:positivity}; permits arbitrary excess components in the kernel incidence \\
Torsion-free over a DVR & A torsion-free module over a DVR is flat & Lemma~\ref{lem:family}; turns domination of reduced components into flatness \\
Euler characteristic and numerical intersections & Proper flatness keeps every twist's Euler characteristic constant; the top coefficients give multiplicity-weighted intersection numbers & Lemma~\ref{lem:conservation}; apply both twisting variables to conserve every multidegree \\
Hurwitz / argument principle & A simple zero of a holomorphic function persists under uniform perturbation & Lemma~\ref{lem:hurwitz}; moves smooth Fermat zeros to nearby fibres \\
Projective duality & A smooth degree-$e$ hypersurface has dual degree $e(e-1)^{N-1}$ & Lemma~\ref{lem:dualdeg} and Proposition~\ref{prop:coneshift}; computes the target-side count \\
\bottomrule
\end{tabularx}
\end{center}

\section{Proof audit: no hidden transfer step}\label{app:audit}

The following checklist gives a direct route for verifying the logical closure
of the argument.

\begin{enumerate}[leftmargin=2.2em,itemsep=3pt]
\item Starting from $\dcb(F_n)\le m$, Lemma~\ref{lem:normalform} produces a
regular polynomial family.  It does not require a regular family of matrices.
\item Lemma~\ref{lem:fiberdet} applies Theorem~\ref{thm:detlemma} to each
chosen generic fibre separately; hence the absence of a coherent matrix family
causes no gap.
\item Degree-$m$ homogenization gives the exact determinant identity and the
special factorization $x_0^{m-n}\widetilde F_n$.
\item The open Gauss graph uses only points with nonzero differential, so its
fibre closure is exactly the multiplicity-one cycle $\Gamma_c$.
\item Vertical components are removed before closure.  Every retained
component dominates the curve, making the reduced total space flat.
\item Generic fibre coefficients are one because the total space is smooth at
a dense point of each generic fibre component.
\item The special fibre is effective, and proper flat specialization conserves
its multidegrees.
\item Lemma~\ref{lem:hurwitz} supplies actual open-graph points on nearby
fibres; it does not merely produce formal limiting covectors.
\item Nonvanishing differential forces those nearby points onto a unique
multiplicity-one branch.
\item Their components dominate the curve because the parameters lie in $U$;
therefore their limits lie in $W_0$.
\item Density and equal dimension make $\Con(X)$ a component of $[W_0]$ with
positive coefficient.
\item Effectivity yields
$\pdeg_{n-2}(\Con(X))\le\pdeg_{n-2}([W_0])$, and conservation identifies the
right side with $\pdeg_{n-2}(\Gamma_c)$.
\item Theorem~\ref{thm:detlemma} bounds that generic-fibre multidegree without
smoothness or reducedness assumptions.
\item Proposition~\ref{prop:coneshift} and Lemma~\ref{lem:dualdeg} compute the
special-side value, and Lemma~\ref{lem:estimates} completes the arithmetic.
\end{enumerate}

\begin{acks}
The author used the assistance of large language models in all stages of the effort. 
\end{acks}

\begingroup
\def\UrlFont{\small\ttfamily}
\setlength{\LTpre}{0.5\baselineskip}
\setlength{\LTpost}{0.5\baselineskip}
\section{End-to-end Lean verification and reproduction guide}
\label{app:lean-verification}

This appendix describes the accompanying Lean development, identifies the
precise statements that it proves, and gives instructions for rebuilding and
inspecting the proofs.  The author of the development is Karthik Sheshadri.  The sources and the
archived release, which also contains this paper, are at
\begin{itemize}[leftmargin=*]
\item Lean sources: \leanRepositoryURL;
\item archived release: \leanArchiveDOI.
\end{itemize}
The entry point is \nolinkurl{OriginalConormal.lean}; the final numerical
theorems are in \nolinkurl{Sheshadri/ConormalExactMultiplicity.lean} and
\nolinkurl{Sheshadri/ConormalBorderMultiplicity.lean}.

\begin{readerbox}{What ``end to end'' means here}
For every pair of natural numbers $n,m$ with $n\geq3$, the development proves
the ordinary and symmetric lower bounds from a literal determinant
representation, or from literal border membership, respectively.  The final
theorems do not ask the user to provide a B\'ezout theorem, a generic flag,
a degree formula, a finite section, a specialization principle, or a geometric
certificate.  These inputs are proved or constructed in the dependency chain.

The formal implementation uses \emph{local Gauss sections and their full
intersection lengths} to obtain the numerical consequences of
conormal/Gauss specialization.  It is not a claim that every intermediate
global-cycle assertion in the main text has been formalized verbatim.
Section~\ref{app:lean-scope} specifies this distinction.
\end{readerbox}

\subsection{The statements checked by Lean}
\label{app:lean-statements}

Throughout the development, the target is the actual multivariate polynomial
\[
 f_n=\sum_{i\in\{0,\ldots,n-1\}}X_i^n
 \quad\text{in}\quad \CC[X_0,\ldots,X_{n-1}].
\]
This indexing differs from $F_n=x_1^n+\cdots+x_n^n$ only by renaming variables.
An affine matrix pencil is an $(n+1)$-tuple of constant $m\times m$ matrices,
and its polynomial is literally
\[
 \det\left(A_0+\sum_{i=0}^{n-1}X_iA_{i+1}\right).
\]
In the symmetric model every coefficient matrix is required to equal its
transpose.  The definitions occur in \nolinkurl{PolynomialModel.lean}, in
the namespace \nolinkurl{Sheshadri.PolynomialModel}.

Put
\begin{align*}
 C(n)&=n(n-1)^{n-2},\\
 B(m,n)&=\sum_{i=1}^{n-1}
   \binom mi\binom{m-1}{n-1-i}\binom{n-2}{i-1},\\
 B_{\mathrm{sym}}(m,n)&=2^{n-2}\binom m{n-1}.
\end{align*}
Their Lean names are, respectively, \nolinkurl{fermatCount},
\nolinkurl{ordinaryBezout}, and \nolinkurl{symmetricBezout}.  These are
natural-number definitions, including the usual zero values of binomial
coefficients outside their range.  The final real-number inequalities use
\texttt{\detokenize{Real.exp 1}} for $e$.

The following is a complete account of the inputs of the final endpoints.
All entries quantify only $n,m\in\mathbb N$, have the additional hypothesis
$3\leq n$, and have no further implicit geometric or typeclass parameter.
The short names in the first table are in the namespace
\nolinkurl{Sheshadri.ConormalExactMultiplicity}.

\begingroup\small
\begin{longtable}{@{}>{\raggedright\arraybackslash}p{0.36\textwidth}>{\raggedright\arraybackslash}p{0.40\textwidth}>{\raggedright\arraybackslash}p{0.18\textwidth}@{}}
\toprule
Lean theorem & Remaining representation hypothesis & Conclusion\\
\midrule
\endhead
\nolinkurl{ordinary_exact_count} &
\nolinkurl{ExactRepresentation}\newline\texttt{n m Complex} & $C(n)\leq B(m,n)$\\[2pt]
\nolinkurl{symmetric_exact_count} &
\nolinkurl{SymmetricExactRepresentation}\newline\texttt{n m Complex} &
$C(n)\leq B_{\mathrm{sym}}(m,n)$\\[2pt]
\nolinkurl{ordinary_exact_quadratic} &
\nolinkurl{ExactRepresentation}\newline\texttt{n m Complex} &
$\displaystyle\frac{(n-1)^2}{4e}\leq m$\\[5pt]
\nolinkurl{symmetric_exact_quadratic} &
\nolinkurl{SymmetricExactRepresentation}\newline\texttt{n m Complex} &
$\displaystyle\frac{(n-1)^2}{2e}\leq m$\\
\bottomrule
\end{longtable}
\endgroup

The border names below are in
\nolinkurl{Sheshadri.ConormalBorderMultiplicity}.  In each paired row,
``ordinary'' and ``symmetric'' select the corresponding hypothesis and
conclusion.

\begingroup\small
\begin{longtable}{@{}>{\raggedright\arraybackslash}p{0.38\textwidth}>{\raggedright\arraybackslash}p{0.38\textwidth}>{\raggedright\arraybackslash}p{0.18\textwidth}@{}}
\toprule
Lean theorem & Remaining border hypothesis & Conclusion\\
\midrule
\endhead
\nolinkurl{ordinary_border_count},
\nolinkurl{symmetric_border_count} &
\nolinkurl{OrdinaryFiniteSupportBorder} \texttt{n m},
\nolinkurl{SymmetricFiniteSupportBorder} \texttt{n m} &
$C(n)\leq B(m,n)$ or $B_{\mathrm{sym}}(m,n)$\\[4pt]
\nolinkurl{ordinary_border_quadratic},
\nolinkurl{symmetric_border_quadratic} &
The same finite-support predicates &
$\displaystyle\frac{(n-1)^2}{4e}\leq m$ or
$\displaystyle\frac{(n-1)^2}{2e}\leq m$\\[7pt]
\nolinkurl{ordinary_borderAtSize_count},
\nolinkurl{symmetric_borderAtSize_count} &
\nolinkurl{OrdinaryBorderAtSize} \texttt{n m},
\nolinkurl{SymmetricBorderAtSize} \texttt{n m} &
The corresponding count bound\\[4pt]
\nolinkurl{ordinary_borderAtSize_quadratic},
\nolinkurl{symmetric_borderAtSize_quadratic} &
The same canonical-support predicates &
The corresponding quadratic bound\\[4pt]
\nolinkurl{ordinary_zariski_count},
\nolinkurl{symmetric_zariski_count} &
\nolinkurl{OrdinaryZariskiBorder} \texttt{n m},
\nolinkurl{SymmetricZariskiBorder} \texttt{n m} &
The corresponding count bound\\[4pt]
\nolinkurl{ordinary_zariski_quadratic},
\nolinkurl{symmetric_zariski_quadratic} &
The same Zariski-border predicates &
The corresponding quadratic bound\\
\bottomrule
\end{longtable}
\endgroup

For example, the full mathematical content of the ordinary finite-support
border endpoint is exactly
\[
 \forall n,m\in\mathbb N,\quad
 3\leq n\ \Longrightarrow\
 \operatorname{OrdinaryFiniteSupportBorder}(n,m)
 \ \Longrightarrow\
 \frac{((n:\mathbb R)-1)^2}{4\exp(1)}\leq(m:\mathbb R).
\]
Positive matrix size is derived inside the proof.  The stronger count
theorems are retained rather than only their asymptotic consequences.
The direct Zariski quadratic endpoints compose the Zariski count theorems
with \nolinkurl{Sheshadri.ordinary_quadratic_bound} and
\nolinkurl{Sheshadri.symmetric_quadratic_bound}.  They do not
need a Zariski-to-Euclidean closure theorem.

There is also an exact, non-Fermat version.  The theorems
\nolinkurl{ordinary_smooth_degree_bound} and
\nolinkurl{symmetric_smooth_degree_bound} assume $n\geq2$, a nonzero
homogeneous $f$ of positive degree $d$, a literal determinant representation,
and
\[
 x\ne0,\quad f(x)=0\quad\Longrightarrow\quad
 \text{some }\partial_i f(x)\ne0.
\]
They conclude $d(d-1)^{n-2}\leq B(m,n)$ or
$d(d-1)^{n-2}\leq B_{\mathrm{sym}}(m,n)$.  The smoothness and degree
hypotheses are proved for $f_n$ before applying these results.

\subsection{The closure definitions are literal}
\label{app:lean-closure}

Let $S$ be a finite set of exponent vectors.  The coefficient space used in
\nolinkurl{BorderRepresentation.lean} is the finite-dimensional complex
space $\CC^S$, equipped with its usual topology.  A coefficient vector $c$
denotes the polynomial $\sum_{\alpha\in S}c_\alpha X^\alpha$.  The ordinary
locus consists of exactly those vectors whose polynomial is a size-$m$
affine determinant; the symmetric locus additionally requires symmetric
coefficient matrices.  Thus
\texttt{\detokenize{OrdinaryFiniteSupportBorder n m}} means that there exists a finite
$S$ and a vector representing $f_n$ that belongs to the Euclidean closure
of this actual locus.  No family or specialization assertion is part of
this definition.

The canonical \nolinkurl{AtSize} predicates use
\texttt{\detokenize{coefficientSupport n m}}.  The development proves that this
finite set contains the support of $f_n$ and of every size-$m$ affine
determinant.  The theorems \nolinkurl{ordinaryLocus_eq_range} and
\nolinkurl{symmetricLocus_eq_range} identify the loci with the coefficient
images of all such pencils.  Consequently, the canonical support introduces
no unproved restriction on the determinant locus.

The Zariski predicates, defined in \nolinkurl{ZariskiBridge.lean}, use the
common zeros of all polynomials vanishing on the actual coefficient locus.
The proof needed for the Euclidean endpoints is only
\[
 \overline{T}^{\,\mathrm{Euclidean}}
 \ \subseteq\ \overline{T}^{\,\mathrm{Zariski}}.
\]
It is proved using continuity of polynomial evaluation.  The finite-support
coefficient reindexing required to apply this inclusion is also proved.
The final argument then constructs its algebraic family directly from
Zariski membership.  It neither assumes that these two closures agree for
arbitrary sets nor uses their reverse inclusion.  An older conditional
interface named \nolinkurl{UniversalPolynomialImageClosure} is retained in
the source, but it is not an input to any final endpoint listed above.

\subsection{A definition-level dictionary}
\label{app:lean-dictionary}

The formal objects below are concrete algebraic objects, rather than names
for desired numerical properties.  Local module names in this appendix
refer to files in the \nolinkurl{Sheshadri/} directory.

\begingroup\small
\begin{longtable}{@{}>{\raggedright\arraybackslash}p{0.30\textwidth}>{\raggedright\arraybackslash}p{0.64\textwidth}@{}}
\toprule
Mathematical object & Formal realization\\
\midrule
\endhead
Affine determinant and target &
\nolinkurl{PolynomialModel.affinePencil},
\nolinkurl{determinant}, and \nolinkurl{fermat}; actual matrices and
multivariate polynomials.\\[4pt]
Polar section &
\nolinkurl{SmoothGaussSectionLength.polarChartIdeal}; the ideal generated
by $f$, the selected derivative combinations, and $a\cdot x-1$.\\[4pt]
Local intersection algebra &
\nolinkurl{IsolatedLocalMultiplicity.LocalIntersectionRing}; the actual
polynomial quotient localized at the evaluation point.\\[4pt]
Multiplicity &
\nolinkurl{IsolatedLocalMultiplicity.multiplicity}; module length of that
local ring over itself, with proved equality to field dimension in the
rational Artinian setting.\\[4pt]
Special Gauss section &
\nolinkurl{ConormalGaussSection.Section}; actual flags and distinct points,
vanishing, finite quotient, Artinian local rings, smoothness, and a proved
length sum.  \nolinkurl{exists_section} constructs the record.\\[4pt]
Border family &
\nolinkurl{DeterminantDvrFamily.Family}; a DVR, residue and generic maps,
a polynomial over the DVR, one generic pencil, a uniformizer, and literal
specialization and determinant identities.\\[4pt]
Open-locus saturation &
\nolinkurl{HomogeneousSystemSaturation.saturation}; the ideal of $q$ such
that, for every displayed $g_i$, some power $g_i^{N_i}q$ belongs to the
original ideal.\\[4pt]
Finite generic section &
\nolinkurl{ActualGenericGaussSection.exists_actual_section}; constructs
flags with prescribed residues and a finite quotient by the derivative
saturation of the actual generic polar equations.\\[4pt]
Incidence multiplicity &
The actual normalized kernel-incidence equation quotient in a valid
projective chart, localized at its evaluation point; identified with the
polar local ring by proved ring equivalences.\\
\bottomrule
\end{longtable}
\endgroup

For clarity, write the special affine polar ideal as
\[
 I_0=\left(f_n,\
 \sum_i\ell_{1i}\partial_i f_n,\ldots,
 \sum_i\ell_{n-2,i}\partial_i f_n,\ a\cdot x-1\right).
\]
It has $n$ displayed equations in $n$ affine variables.  If $p$ is one of
its selected points, its contribution is
\[
 \mu_p=\operatorname{length}_{(\CC[x]/I_0)_p}
                  (\CC[x]/I_0)_p,
\]
not an assigned weight and not merely the number one.  This distinction
is essential when the polar section is nonreduced.

\subsection{Step 1: construct the special Gauss section and its degree}
\label{app:lean-special}

The construction begins with a homogeneous polynomial that is smooth away
from the cone vertex.  \nolinkurl{SmoothGaussPolars} proves the relevant
Euler identity and gradient basepoint-freeness.  The modules
\nolinkurl{GenericGaussPolars}, \nolinkurl{IrrelevantSaturation}, and
\nolinkurl{GenericGaussPolarSequence} choose successive polar combinations
by avoiding finitely many associated primes.  Irrelevant saturation preserves
sufficiently high graded pieces, while the chosen equations act regularly
on the appropriate saturated quotients.  Generic regularity is thus the
output of an existence proof, not a premise supplied by the caller.

\nolinkurl{SmoothGaussDegree} uses exact sequences of actual graded
quotients to calculate the eventual Hilbert value after the $n-2$ polar
cuts.  It obtains $d(d-1)^{n-2}$.  The construction then chooses a linear
form avoiding the resulting projective support.  Setting this form to one
gives an actual finite-dimensional affine algebra, as proved in
\nolinkurl{SmoothGaussSection} and \nolinkurl{SmoothGaussSectionLength}.

The modules \nolinkurl{ArtinianLocalDecomposition},
\nolinkurl{FiniteQuotientPointFamily}, and
\nolinkurl{SmoothGaussPointFamily} identify the dimension of this finite
algebra with the sum of the full local lengths at its distinct support
points.  The resulting theorem
\nolinkurl{ConormalGaussSection.exists_section} supplies all the data in
the \nolinkurl{Section} record and the identity
\[
 \sum_p\mu_p=d(d-1)^{n-2}.
\]
Specializing to the actual Fermat polynomial gives $\sum_p\mu_p=C(n)$.
The endpoint proof invokes this constructor and discharges its hypotheses;
it does not assume the existence or the numerical degree of such a section.

The chart condition $a\cdot x=1$ is chosen after the homogeneous polar
section.  The proof does not fix an original Fermat coordinate to one and
then explicitly solve the resulting polar equations to obtain the lower
count.  The degree comes from the regular-cut Hilbert calculation.

\subsection{Step 2: construct a genuine DVR family from the border}
\label{app:lean-family}

The border argument starts with the coefficient-space membership of
Section~\ref{app:lean-closure}.  \nolinkurl{AlgebraicCurveSelection} and
\nolinkurl{CurveThroughBoundary} construct a curve through the boundary
point of the polynomial-image closure.  \nolinkurl{NormalGoingDown},
\nolinkurl{CurveNormalization}, and \nolinkurl{FiniteCurveExtension}
construct the normalized branch and the finite extension needed to carry
a coherent generic preimage.  Localizing at the selected point supplies
a discrete valuation ring $D$, a residue map $D\to\CC$, and its actual
fraction field $L$.

\nolinkurl{GenericPolynomialLift} and \nolinkurl{AlgebraicBorderFamily}
produce one tuple of generic matrix parameters.  They do not choose
unrelated representations at separate fibres.  Pole clearing in
\nolinkurl{DvrRationalEntries} uses one common power of a uniformizer.
\nolinkurl{DeterminantDvrFamily} then reconstructs a polynomial
$F\in D[x_1,\ldots,x_n]$ and a pencil $A$ over $L$ satisfying the actual
polynomial identities
\[
 F\bmod\pi=f_n,\qquad F\otimes_D L=\det A(x).
\]
The theorem \nolinkurl{Family.isFractionRing} verifies that the field used
by the pencil is the fraction field of the constructed DVR.  Symmetry is
retained in the symmetric construction.

The final family-selection results are
\nolinkurl{ordinary_border_exact_or_family} and
\nolinkurl{symmetric_border_exact_or_family}.  They return either an exact
representation, to which the exact theorem applies, or the family just
described.  An arbitrary algebraic curve is never replaced by an assumed
parametrization over $\CC[t]$, and the generic field $L$ is never identified
with $\CC$.

\subsection{Step 3: prove that a compatible generic section is finite}
\label{app:lean-generic}

It would be incorrect to infer finite generic fibre merely from finite
special fibre: components can escape the special fibre.  The development
therefore proves generic finiteness separately, using the actual family.

\nolinkurl{DeterminantDvrHomogeneous} homogenizes $F$ in an added variable
and proves the dehomogenization and derivative identities.  Its theorem
\nolinkurl{generic_hypersurface_hilbert} bounds the degree of the initial
Hilbert polynomial by $n-1$.  The saturated generic construction works on
the open locus where the added coordinate and at least one spatial
derivative are nonzero.

The prescribed special flags cannot simply be replaced by unrelated
generic flags.  \nolinkurl{RegularFlagsWithResidue.exists_regular_coefficients}
chooses coefficients in $D$ with exactly the required residues, while
avoiding the finitely many associated primes of the generic quotients.
The choices are perturbations by uniformizer multiples.  The modules
\nolinkurl{PrescribedResidueGaussSequence} and
\nolinkurl{FinitePrescribedGaussSection} apply this construction to all
$n-2$ polar cuts and the final primal cut.  Normalization of the primal
equation divides only by a proved unit of $D$ and preserves its prescribed
residue.

\nolinkurl{HomogeneousHilbertDegree} proves the degree decrease under a
regular homogeneous cut and the required monotonicity under enlargement
of the ideal.  \nolinkurl{SaturatedSectionHilbert} applies these results
through the actual cut sequence.  After $n-1$ cuts the Hilbert function is
eventually constant.  Regularity of the added coordinate then proves that
the corresponding affine chart algebra is finite-dimensional.

Dehomogenization identifies an actual ideal $J\subseteq L[x]$.  If
\[
 I_\eta=\left(\det A,\
 \sum_i\widetilde\ell_{1i}\partial_i\det A,\ldots,
 \sum_i\widetilde\ell_{n-2,i}\partial_i\det A,\
 \widetilde a\cdot x-1\right),
\]
then \nolinkurl{ActualGenericGaussSection.exists_actual_section} proves
\begin{gather*}
 \dim_L L[x]/J<\infty,\\
 I_\eta\subseteq J\subseteq
 \operatorname{Sat}_{\partial\det A}(I_\eta),\qquad
 J=\operatorname{Sat}_{\partial\det A}(I_\eta).
\end{gather*}
Here $q\in\operatorname{Sat}_{g}(I)$ means
$\forall i\,\exists N_i\geq0:\ g_i^{N_i}q\in I$.  This operation removes
components supported on the excluded closed locus; it does not take a
radical.  Where one $g_i$ is a unit, the localized saturated and original
ideals are proved equal.

Finally, \nolinkurl{SaturatedChartSmoothness} constructs a derivative
combination whose class in the full finite quotient $L[x]/J$ is a unit.
Regularity gives injectivity of multiplication; finite dimension gives
surjectivity and hence the inverse.  Therefore all geometric support
points of $J$ are smooth points of the determinant hypersurface, including
after passage to an algebraic closure.  The full quotient by $I_\eta$
need not be finite: excess components outside this smooth section are
allowed throughout.

\subsection{Step 4: preserve singular local multiplicities}
\label{app:lean-persistence}

The lower comparison must preserve nilpotents, not merely show that the
selected special points are in a limiting support.  This is the role of
\nolinkurl{DVRRootPersistence}.

At a selected special point, the actual local polar algebra is Artinian
and is given by a square equation system.  The regular-sequence machinery
in \nolinkurl{CoordinatePowerRegular},
\nolinkurl{LocalCompleteIntersection}, and
\nolinkurl{IsolatedCompleteIntersection} derives regularity of these
equations in the local polynomial ring.  It does not require a nonzero
Jacobian determinant.  The total local equation quotient then has regular
uniformizer action.  Cancelling generic denominators shows that contraction
from the generic fibre annihilates the entire original special local
algebra.

\nolinkurl{DVRSaturatedPersistence} handles the additional saturation in
$J$.  At a selected special point, one spatial derivative has nonzero
value, and consequently is a unit in its full Artinian local ring.  Power
membership in the saturation, followed by cancellation of this unit, shows
that every contracted element of $J$ still annihilates that ring.
\nolinkurl{DVRSaturatedGaussPersistence} applies this result to the actual
family and the residue-preserving flags.  Its theorem
\nolinkurl{section_saturated_contraction_local_eq} proves equality of the
corresponding localized ideals.  Thus this step retains local scheme
structure and local length.

Let $T=J\cap D[x]$ and let $T_0$ be its residue ideal.  Simultaneous
interpolation onto the full special local algebras gives
\[
 \sum_p\mu_p\ \leq\ \dim_{\CC}\CC[x]/T_0.
\]
The quotient $D[x]/T$ embeds in its generic fibre and is torsion-free over
$D$.  \nolinkurl{DVRFibreDimension} proves the residue-dimension inequality
\[
 \dim_{\CC}\CC[x]/T_0\ \leq\ \dim_L L[x]/J.
\]
It does not assume finite generation of the whole horizontal algebra over
$D$; the already proved finite generic dimension is used where needed.
\nolinkurl{DVRSectionLength.selected_lengths_le_generic} combines these
facts, and \nolinkurl{DVRGaussLengthBound.family_section_degree_le}
instantiates them to obtain
\begin{equation}
 \label{eq:lean-special-generic}
 n(n-1)^{n-2}\ \leq\ \dim_L L[x]/J.
\end{equation}

\subsection{Step 5: B\'ezout with isolated singular zeros and excess}
\label{app:lean-bezout}

The precise B\'ezout statement matters.  The development proves an
isolated-intersection inequality; it does not assert that an arbitrary
nonproper intersection has the expected total length.  The theorem
\nolinkurl{FieldSingularBezoutWithExcess.isolated_multiplicity_bound}
is over a characteristic-zero field and has the following inputs:
\begin{enumerate}[leftmargin=*]
\item Three finite nonempty homogeneous coordinate blocks of projective
dimensions $a,b,c$, and exactly $a+b+c$ actual equations.
\item The displayed multihomogeneity of each equation, with multidegree
$(d_{e,1},d_{e,2},d_{e,3})$.
\item A finite family of common zeros, each with a specified nonzero chart
coordinate in every block, distinct modulo the independent projective
block scalings.
\item Artinianness of each selected point's actual localized chart
equation ring.
\end{enumerate}
Its conclusion is
\begin{equation}
 \label{eq:lean-bezout}
 \sum_p\operatorname{length}(\mathcal A_p)
 \ \leq\
 [H_1^aH_2^bH_3^c]
 \prod_e(d_{e,1}H_1+d_{e,2}H_2+d_{e,3}H_3).
\end{equation}
The family of selected zeros need not exhaust the zero set.  Other
components may have positive dimension.  The selected roots need not be
simple or reduced.  A companion theorem,
\nolinkurl{isolated_multiplicity_bound_of_isOpen}, derives Artinianness
from literal isolatedness expressed as an open singleton in the actual
chart spectrum.

The proof of~\eqref{eq:lean-bezout} is itself part of the formalization.
It interpolates the given equations with a split comparison system of the
same multidegrees, contracts the generic ideal without radicalizing it,
and proves flatness and conservation of actual finite multigraded pieces.
Local complete-intersection persistence retains all selected local
lengths at the original endpoint.  At the comparison endpoint, the split
equation quotient surjects onto the contracted-fibre quotient, giving the
upper-bound inequality in the correct direction.

The split-system estimate uses a Hilbert-function inequality for products,
not an unproved regular-sequence equality in the entire homogeneous
coordinate ring.  Expanding products gives finitely many branches.
Branches selecting too many linear equations in a block vanish in
sufficiently positive degree; branches meeting exactly the dimension
quotas contribute at most one.  The number of quota choices is the
coefficient on the right of~\eqref{eq:lean-bezout}.  All required degree
thresholds and componentwise shifts are proved.  This approach allows
irrelevant and excess components without identifying point counts with
the degree of an arbitrary homogeneous quotient.

The main dependency modules for this part are
\nolinkurl{FieldSingularBezoutDeformationBound},
\nolinkurl{FieldBezoutSplitHilbert},
\nolinkurl{FieldBezoutVariableTransport}, and
\nolinkurl{FieldSingularBezoutWithExcess}, together with their local
persistence and multigraded algebra dependencies.  The comparison split
system is a general B\'ezout upper-bound construction; it is unrelated to
a fixed-coordinate explicit count for the Fermat polynomial.

\subsection{Step 6: apply B\'ezout to the actual determinant incidence}
\label{app:lean-incidence}

At a smooth determinant zero, the proof constructs normalized kernel
vectors and valid cofactor charts.  In the ordinary model it also
constructs one reduction matrix valid simultaneously for the finite
selected family.  The Schur-complement calculations are equalities of
ideals and actual quotient/local-ring equivalences, not only equalities
of their zero sets.  The symmetric model uses its native one-kernel-vector
incidence rather than discarding the symmetry condition.

\nolinkurl{FieldPolarMultiplicityBound} applies the singular isolated-zero
B\'ezout theorem to these literal incidence equations.  It proves the
equation count, multidegrees, projective distinctness and Artinian local
conditions from the constructed polar data.  The local equivalences
transfer the full intersection lengths.  The resulting incidence
coefficients are precisely $B(m,n)$ and $B_{\mathrm{sym}}(m,n)$.

The generic field of the border family need not be algebraically closed.
\nolinkurl{ArbitraryFieldFinitePolarBound} therefore extends the actual
finite algebra $L[x]/J$ to an algebraic closure $\overline L$.  Explicit
tensor/quotient equivalences prove
\[
 \dim_{\overline L}\bigl(\overline L\otimes_L L[x]/J\bigr)
 =\dim_L L[x]/J.
\]
The unit derivative combination remains a unit after extension, so every
support point over $\overline L$ is smooth.  The proved saturation
inclusions suffice to identify each localized $J$-algebra with the
original polar algebra.  The argument uses the required inclusion under
mapping an ideal; it does not assume that saturation commutes with every
base change.

\nolinkurl{FieldFinitePolarBound} constructs the complete point family
of the finite quotient and sums its local lengths.  Consequently
\nolinkurl{ArbitraryFieldFinitePolarBound.ordinary_finite_saturated_polar_bound}
and its symmetric counterpart give
\begin{equation}
 \label{eq:lean-generic-upper}
 \dim_L L[x]/J\leq B(m,n),
 \qquad\text{or}\qquad
 \dim_L L[x]/J\leq B_{\mathrm{sym}}(m,n).
\end{equation}
No generic homogeneity of the original affine determinant polynomial and
no finite-dimensionality of its entire unsaturated polar quotient is
assumed here.

\subsection{Step 7: assemble the endpoints and the constants}
\label{app:lean-assembly}

The family-level theorem
\nolinkurl{ConormalBorderMultiplicity.ordinary_family_count} constructs
the special section and the compatible finite generic section, applies
\eqref{eq:lean-special-generic} and~\eqref{eq:lean-generic-upper}, and
concludes $C(n)\leq B(m,n)$.  The symmetric theorem performs the same
assembly with the symmetric incidence.  The Zariski endpoints invoke
the exact-or-family constructor, and the Euclidean endpoints invoke the
proved one-way closure transfer.  Thus the intermediate family and section
records are outputs of proved constructors, not undeclared obligations
of the endpoint user.

\nolinkurl{Bezout.lean} proves the coefficient identities and the
Vandermonde estimate.  \nolinkurl{Numerical.lean} proves
\[
 B(m,n)\leq\left(\frac{4em}{n-1}\right)^{n-1},
 \qquad
 B_{\mathrm{sym}}(m,n)
 \leq\left(\frac{2em}{n-1}\right)^{n-1}.
\]
Combining these with $(n-1)^{n-1}\leq C(n)$ and the positivity following
from $n\geq3$ yields the two quadratic bounds.  All casts between natural
numbers and real numbers, the exponential constants, and the power
comparison are checked by Lean.

The exact proof is shorter: it constructs the special section, identifies
its actual local rings with the kernel incidence, and applies B\'ezout
directly.  It does not need the border family or specialization comparison.

\subsection{What the formalization establishes, and what it does not}
\label{app:lean-scope}

The verified numerical chain for the border theorem is
\[
 \underbrace{n(n-1)^{n-2}}_{\text{constructed special local lengths}}
 \ \leq\
 \underbrace{\dim_L L[x]/J}_{\text{constructed generic section}}
 \ \leq\
 \underbrace{B(m,n)}_{\text{proved incidence B\'ezout coefficient}},
\]
with the symmetric replacement in the last term.  This is a local
Gauss-section implementation of the conormal/Gauss specialization
argument.  It establishes the stated exact and border lower-bound
endpoints without assuming a global projective cycle/intersection
certificate.

The written proof also distinguishes numerical
constancy of fibre multidegrees from rational equivalence of fibres over
an arbitrary base curve: the relevant numerical constancy is obtained
from the bigraded Euler--Hilbert calculation.  Likewise, the cone
comparison is made away from the vertex, and the primal hyperplane is
chosen to avoid the vertex.  No global projective-line-bundle description
through the vertex is needed.  These distinctions are consistent with
the local-section implementation; the end-to-end Lean guarantee remains
the precise endpoint statements in Section~\ref{app:lean-statements}.

The development also contains separately proved results on actual Gauss
families, flatness, conormal containment, grading, Hilbert functions, and
projective charts.  It does not claim a separately complete global Chow
ring library, the full global projective fibre-cycle assembly of the main
text, or an identification of every concrete local chart ring with a
stalk on a product of Mathlib \nolinkurl{Proj} schemes.  Those broader
claims are not needed as hypotheses of the assembled endpoints.

Older modules retain conditional interfaces called
\nolinkurl{GaussDegreeTheory} and \nolinkurl{FamilyCertificate}.  A theorem
accepting one of these records remains conditional.  Importing its module
does not make the record an axiom, but a clean axiom list also does not
prove the record exists.  The final \nolinkurl{Multiplicity} endpoint
modules use the constructed local-section route described above and do
not require those interfaces.

The only affine chart used in the generic construction sets the
\emph{added homogenizing coordinate} $x_0$ to one.  No original coordinate
of $F_n$ is set to one, and no explicit count of polar points is used: the
lower count comes from the Hilbert-function calculation of
Section~\ref{app:lean-special}.

\subsection{What the audits check}
\label{app:lean-audits}

The release's entry-point dependency graph contains 225 local modules.
The principal axiom audit requests 349 reports.  The accompanying
\nolinkurl{BUILD_LOG.txt}, \nolinkurl{AXIOM_AUDIT.txt}, and
\nolinkurl{IMPORT_AUDIT.txt} record successful build and audit results.
The scripts recompute the checks from the delivered source; these counts
are descriptive release metadata rather than mathematical assumptions.
The implementation and source-level review were AI-assisted.  The
semantic audit documents record separate automated source reviews, not
independent human peer review.  The compiler and axiom transcripts are
the reproducible machine-checking evidence.

There are three distinct checks.
\begin{enumerate}[leftmargin=*]
\item \textbf{Compilation.}  Lean elaborates the local development against
the pinned dependencies and checks the resulting proof terms.  A build
failure is not converted into a successful verification claim.
\item \textbf{Foundational dependencies.}
\nolinkurl{AxiomAudit.lean} runs \texttt{\detokenize{#print axioms}} on the principal
declarations, including the final endpoints.  The only reported axioms
are Lean's standard \nolinkurl{propext}, \nolinkurl{Classical.choice}, and
\nolinkurl{Quot.sound}.  \nolinkurl{audit_axioms.py} checks that every
requested report is present and rejects unexpected dependencies.
\nolinkurl{audit_imports.py} checks the local entry-point import graph for
missing files, custom axioms, unfinished proofs, unsafe shortcuts and the
excluded alternative explicit-slice modules.
\item \textbf{Mathematical statement and applicability.}
\nolinkurl{END_TO_END_HYPOTHESES_AUDIT.txt} reviews the literal definitions,
the final theorem types, and how the proof supplies every constructor's
applicability conditions.  This check is necessary because a theorem of
the form ``if the desired count bound holds, then the desired conclusion
holds'' can compile with only standard foundational axioms while proving
little about the intended geometric problem.
\end{enumerate}

These checks address different questions.  A source scan is not a
substitute for compilation, an axiom report is not a semantic audit, and
the existence of a successful transcript is not a substitute for rebuilding
the sources when independent reproduction is desired.  The trusted base
includes the Lean kernel and the pinned mathematical definitions and
dependencies.  The verification is a proof for all admitted $n,m$, not
numerical testing of finitely many dimensions.

Construction-specific reading guides are
\nolinkurl{GAUSS_DEGREE_CONSTRUCTION_AUDIT.txt},
\nolinkurl{BORDER_DVR_CONSTRUCTION_AUDIT.txt},
\nolinkurl{GENERIC_GAUSS_SECTION_AUDIT.txt},
\nolinkurl{DVR_LOCAL_PERSISTENCE_AUDIT.txt}, and
\nolinkurl{EXACT_MULTIPLICITY_AUDIT.txt}.  Historical audit files are
retained to document earlier stages; their earlier lists of remaining
obligations are superseded by \nolinkurl{ORIGINAL_CONORMAL_SCOPE.txt} and
\nolinkurl{END_TO_END_HYPOTHESES_AUDIT.txt}.

\subsection{Reproducing the build from the release sources}
\label{app:lean-reproduce}

The release is source-complete for its local entry-point graph.  It does
not vendor the Lean compiler, Mathlib, or precompiled Mathlib caches.
A fresh machine therefore needs an Internet connection for the initial
toolchain and dependency download, Git, Python~3, a shell capable of
running the supplied Bash script, and Lean's toolchain manager Elan.
The mathematical toolchain is pinned to
\begin{itemize}[leftmargin=*]
\item Lean \nolinkurl{leanprover/lean4:v4.19.0};
\item Mathlib revision
\nolinkurl{c44e0c8ee63ca166450922a373c7409c5d26b00b}
(Mathlib v4.19.0);
\item the remaining dependency revisions in the delivered
\nolinkurl{lake-manifest.json}.
\end{itemize}
Download and extract the archived release, or clone the repository (both linked at the start of this appendix), and keep its manifest.  From the
directory containing \nolinkurl{lean-toolchain}, \nolinkurl{lakefile.toml},
and \nolinkurl{verify.sh}, run:

\begingroup\small
\begin{verbatim}
elan toolchain install leanprover/lean4:v4.19.0
lake env lean --version
lake exe cache get
./verify.sh
\end{verbatim}
\endgroup

If the extraction did not preserve executable permissions, use
\texttt{\detokenize{bash verify.sh}}.  The script executes, in order:

\begingroup\small
\begin{verbatim}
lake build
lake env lean AxiomAudit.lean | tee AXIOM_AUDIT.txt
python3 audit_axioms.py
python3 audit_imports.py
\end{verbatim}
\endgroup

It stops on failure.  A successful return code is zero.  The expected
audit summaries report all requested declarations, only the standard
foundational axioms, and a complete entry-point source graph with no
prohibited shortcuts.  The Mathlib cache is a build accelerator; the
release supplies source proofs for its own modules.  A reader wishing
to rebuild dependencies from source may omit the cache download and
allow the dependency build to complete instead.

The manifest-preserving build above is the reproduction procedure for
the archived release.  \texttt{\detokenize{lake update}} is useful when deliberately
regenerating dependency metadata, but is not necessary for this procedure
and may change transitive dependency resolutions.  If it is used, retain
a copy of the archived manifest, compare every resolved revision with
that copy, and restore the archived manifest before claiming an exact
release reproduction.  The Mathlib checkout can be checked directly:

\begingroup\small
\begin{verbatim}
git -C .lake/packages/mathlib rev-parse HEAD
\end{verbatim}
\endgroup

It must print the Mathlib commit above.  The release's
\nolinkurl{SOURCE_SHA256SUMS.txt} records the delivered file contents
except the checksum manifest itself.  Verify these checksums on the
untouched extraction before rerunning scripts that regenerate logs.
On a system with the GNU checksum utility, run
\texttt{\detokenize{sha256sum -c SOURCE_SHA256SUMS.txt}} from the directory to which
the manifest paths are relative.

\nolinkurl{BUILD_ENVIRONMENT.txt} describes the original build environment
and the small executable-location compatibility wrapper needed by that
container.  The wrapper changes only the process executable-location
lookup and is not part of the mathematical development.  A normal Lean
installation does not need it.  The compiler and shared-library hashes
are recorded there for inspection.

\subsection{A short independent endpoint inspection}
\label{app:lean-smokecheck}

A reader need not begin by reading every implementation file.  After
building the release, create a new file named
\nolinkurl{EndpointInspection.lean} in the project root with the following
contents.  This file only inspects the already defined statements; it
does not introduce assumptions or replace a proof.

\begingroup\small
\begin{verbatim}
import OriginalConormal

set_option pp.explicit true
set_option pp.universes true

open Sheshadri.ConormalExactMultiplicity in
#check @ordinary_exact_quadratic
open Sheshadri.ConormalExactMultiplicity in
#check @symmetric_exact_quadratic

open Sheshadri.ConormalBorderMultiplicity in
#check @ordinary_border_quadratic
open Sheshadri.ConormalBorderMultiplicity in
#check @symmetric_border_quadratic
open Sheshadri.ConormalBorderMultiplicity in
#check @ordinary_borderAtSize_quadratic
open Sheshadri.ConormalBorderMultiplicity in
#check @symmetric_borderAtSize_quadratic
open Sheshadri.ConormalBorderMultiplicity in
#check @ordinary_zariski_count
open Sheshadri.ConormalBorderMultiplicity in
#check @symmetric_zariski_count
open Sheshadri.ConormalBorderMultiplicity in
#check @ordinary_zariski_quadratic
open Sheshadri.ConormalBorderMultiplicity in
#check @symmetric_zariski_quadratic

open Sheshadri.ConormalBorderMultiplicity in
#print axioms ordinary_border_quadratic
open Sheshadri.ConormalBorderMultiplicity in
#print axioms symmetric_border_quadratic

#print Sheshadri.PolynomialModel.ExactRepresentation
#print Sheshadri.PolynomialModel.SymmetricExactRepresentation
open Sheshadri.BorderRepresentation in
#print OrdinaryFiniteSupportBorder
open Sheshadri.BorderRepresentation in
#print OrdinaryBorder
open Sheshadri.BorderRepresentation in
#print ordinaryLocus
open Sheshadri.ZariskiBridge in
#print OrdinaryZariskiBorder
open Sheshadri.ZariskiBridge in
#print zariskiClosure
\end{verbatim}
\endgroup

Run \texttt{\detokenize{lake env lean EndpointInspection.lean}}.  The output makes
all implicit parameters visible.  In particular, the final exact and
border statements must have only the natural numbers $n,m$, the condition
$3\leq n$, and the literal representation or border predicate.  The
ordinary and symmetric endpoint axiom lists must contain no custom axiom
or \nolinkurl{sorryAx}.  Repeat the definition inspection with the symmetric
border names if desired; their extra condition is literal matrix symmetry.

For a mathematical review, a useful order is: inspect the definitions and
final types; read the two endpoint modules; follow
\nolinkurl{exists_section} and \nolinkurl{exists_actual_section}; inspect
\nolinkurl{family_section_degree_le}; then inspect the arbitrary-field
finite polar bound and the singular B\'ezout statement.  At each step,
check both the intermediate theorem's assumptions and the constructor
that supplies them.  This follows the actual proof dependency chain and
separates a reusable conditional lemma from an assembled unconditional
endpoint.

\endgroup

\end{document}